\documentclass{article}
\usepackage{fullpage}%package to reduce margins
\usepackage[usenames]{color}
\usepackage[colorlinks = true]{hyperref}
\usepackage[english]{babel}
\usepackage{bm,bbm,amsmath,amssymb,amsthm,graphicx}
\usepackage{url}
\usepackage[caption=false,font=normalsize,labelfont=sf,textfont=sf]{subfig}
\usepackage[title]{appendix}
\usepackage{algorithm}
\usepackage{algpseudocode}
\usepackage{todonotes}
\usepackage{tikz}
\usetikzlibrary{calc,patterns,decorations.pathmorphing,decorations.markings}

\newcommand{\SL}[1]{\textcolor{black}{#1}}      %orange
\providecommand{\keywords}[1]{\textbf{\textit{Index terms---}} #1}

\numberwithin{equation}{section}

\newcommand{\change}[1]{\ensuremath{\operatorname{#1}}}
\newcommand{\MAT}{\left[ \begin{array}}  %\MAT{ccc}   ... \mat
\newcommand{\mat}{\end{array} \right]}
\newtheorem{Definition}{Definition}[section]

\newtheorem{Lemma}{Lemma}[section]
\newtheorem{Theorem}{Theorem}[section]

\newtheorem{Q}{Question}[section]

\def \st {\operatorname*{s.t. }}

\def \a {\bm{a}}
\def \A {\mathbf{A}}
\def \AA {\mathcal{A}}

\def \b {\bm{b}}
\def \B {\mathbf{B}}
\def \BB {\mathcal{B}}
\def \c {\bm{c}}

\def \CN {\mathcal{CN}}
\def \CCC {\mathbb{C}}

\def \D {\mathbf{D}}

\def \Db {\overline{\mathbf{D}}}
\def \e {\bm{e}}
\def \E {\mathbf{E}}
\def \EE{\mathcal{E}}
\def \EEE{\mathbb{E}}

\def \FF {\mathcal{F}}
\def \g {\bm{g}}

\def \h {\bm{h}}

\def \HH {\mathcal{H}}
\def \I {\mathbf{I}}

\def \K {\mathbf{K}}
\def \KK {\mathcal{K}}

\def \Kt{\widetilde{\mathbf{K}}}
\def \L {\mathbf{L}}

\def \M {\mathbf{M}}

\def \NN {\mathcal{N}}
\def \OO {\mathcal{O}}

\def \PP {\mathcal{P}}
\def \PPP {\mathbb{P}}
\def \q {\bm{q}}
\def \Q {\mathbf{Q}}
\def \QQ {\mathcal{Q}}
\def \QQt {\widetilde{\mathcal{Q}}}
\def \QQb {\overline{\mathcal{Q}}}

\def \r{\bm{r}}

\def \RRR {\mathbb{R}}
\def \S {\mathbf{S}}

\def \Sb {\overline{\mathbf{S}}}
\def \T {\mathbf{T}}
\def \TT {\mathcal{T}}
\def \Th {\widehat{\mathbf{T}}}
\def \u {\bm{u}}
\def \uh {\widehat{\bm{u}}}

\def \v {\bm{v}}
\def \V {\mathbf{V}}

\def \w {\bm{w}}
\def \W {\mathbf{W}}
\def \WW {\mathcal{W}}

\def \x {\bm{x}}
\def \xs {\bm{x}^\star}
\def \xh {\widehat{\bm{x}}}
\def \X {\mathbf{X}}

\def \Xh {\widehat{\mathbf{X}}}
\def \Xs {\mathbf{X}^\star}

\def \y {\bm{y}}

\def \z {\bm{z}}

\def \ZZ {\mathcal{Z}}

\def \sumj{\sum_{j=1}^J}
\def \sumk{\sum_{k=1}^K}
\def \summ{\sum_{m=-2M}^{2M}}

\def \tr{\change{tr}}

\def \balpha {\boldsymbol{\alpha}}

\def \balphab {\overline{\boldsymbol{\alpha}}}
\def \bbeta {\boldsymbol{\beta}}

\def \bbetab {\overline{\boldsymbol{\beta}}}

\def \bmu {\boldsymbol{\mu}}
\def \bmuh {\widehat{\boldsymbol{\mu}}}
\def \bnu {\boldsymbol{\nu}}
\def \bxi {\boldsymbol{\xi}}

\def \sigmat {\widetilde{\sigma}}

\def \deltat {\widetilde{\delta}}
\def \taut {\widetilde{\tau}}
\def \zero {\mathbf{0}}

\begin{document}

\title{Atomic Norm Denoising for Complex Exponentials\\ with Unknown Waveform  Modulations}

\author{Shuang Li, Michael B. Wakin, and Gongguo Tang \thanks{Department of Electrical Engineering, Colorado School of Mines. Email: \{shuangli,mwakin,gtang\}@mines.edu.}}

\date{February 15, 2019 ~~~~Revised: September 04, 2019}

\maketitle

\begin{abstract}
 
Non-stationary blind super-resolution is an extension
%\SL{(Mike: How? Shuang: I think we explained this in the first paragraph of introduction.)}
of the traditional super-resolution problem, which deals with the problem of recovering fine details from coarse measurements. The non-stationary blind super-resolution problem appears in many applications including radar imaging, 3D single-molecule microscopy, computational photography, etc. There is a growing interest in solving non-stationary blind super-resolution task with convex methods due to their robustness to noise and strong theoretical guarantees. Motivated by the recent work on atomic norm minimization in blind inverse problems, we focus here on the signal denoising problem in non-stationary blind super-resolution. In particular, we use an atomic norm regularized least-squares problem to denoise a sum of complex exponentials with unknown waveform modulations. We quantify how the mean square error depends on the noise variance and the true signal parameters. Numerical experiments are also implemented to illustrate the theoretical result.

\end{abstract}

\keywords{Atomic norm denoising, non-stationary blind deconvolution, line spectrum estimation, demodulation, mean
 square error.}

\section{Introduction}
\label{intr}

Super-resolution is the process of recovering high-resolution information of a signal from its coarse-scale measurements~\cite{candes2014towards}. Super-resolution problems arise in a wide variety of applications, including microscopy~\cite{mccutchen1967superresolution}, imaging spectroscopy~\cite{harris1994super}, radar signal demixing~\cite{xie2017radar}, astronomical imaging~\cite{puschmann2005super}, and medical imaging~\cite{greenspan2008super}.
%\note{please add citations to the following sentences}
One specific super-resolution problem has received considerable theoretical study in recent years: recovering positions on the interval $[0,1)$ of unknown point sources given only low-frequency samples~\cite{candes2014towards, candes2013super}. Exchanging the roles of time and frequency, this problem is equivalent to that of line spectral estimation: recovering frequencies on the interval $[0,1)$ of unknown complex exponentials given limited or possibly compressed time-domain samples. The total variation norm has been used to regularize such super-resolution problems, and this is equivalent to the atomic norm that has been used to regularize such line spectral estimation problems.

In this work, we are interested in the non-stationary blind super-resolution scenario, which extends the above super-resolution problem to the setting where each point source is convolved with a unique and unknown point spread function. The term ``non-stationary'' indicates that the point spread functions are potentially different and comes from the field of non-stationary deconvolution~\cite{margrave2011gabor}; the term ``blind'' indicates that the point spread functions are unknown. Non-stationary blind super-resolution problems also appear in applications involving radar imaging~\cite{heckel2016super}, astronomy~\cite{starck2002deconvolution}, photography~\cite{fergus2006removing}, 3D single-molecule microscopy~\cite{quirin2012optimal}, seismology~\cite{margrave2011gabor} and nuclear magnetic resonance (NMR) spectroscopy~\cite{qu2015accelerated, cai2016robust,ying2018vandermonde}.
%\SL{[Check reference~\cite{qu2015accelerated}. It didn't consider the signal model~\eqref{signal_noiseless}. May need to find another place to cite this paper].}
Exchanging the roles of time and frequency, the conventional line spectral estimation problem is modified as follows: each complex exponential is modulated (pointwise multiplied) by an unknown waveform, and this waveform can vary from one complex exponential to the next. Though both problem formulations are equivalent, it is this modified line spectrum estimation problem that we detail in Section~\ref{prob} and refer to throughout this paper.

%Super-resolution is the process of recovering high-resolution information of a signal from its coarse-scale measurements~\cite{candes2014towards}. The super-resolution problems arise in a wide variety of applications, including microscopy~\cite{mccutchen1967superresolution}, imaging spectroscopy~\cite{harris1994super},    radar signal demixing~\cite{xie2017radar}, astronomical imaging~\cite{puschmann2005super}, and medical imaging~\cite{greenspan2008super}.
%The non-stationary blind super-resolution scenario extends traditional super-resolution to address the problem of recovering a sum of unknown complex exponentials from their modulations with unknown waveforms, as is studied in~\cite{yang2016super}. The term ``non-stationary" comes from the field of non-stationary deconvolution~\cite{margrave2011gabor}. \SL{We use ``non-stationary blind" to emphasize that both the modulation waveforms and the point sources contained in the complex exponentials are {\it unknown}, and the {\it unknown} modulation waveforms can {\it vary} with the point sources contained in the {\it unknown} complex exponentials.} Non-stationary blind super-resolution problems also appear in applications involving radar imaging~\cite{heckel2016super}, astronomy~\cite{starck2002deconvolution}, photography~\cite{fergus2006removing}, 3D single-molecule microscopy~\cite{quirin2012optimal}, seismology~\cite{margrave2011gabor} and nuclear magnetic resonance (NMR) spectroscopy~\cite{cai2016robust}.

In recent years, convex methods have been widely used in super-resolution due to their robustness to noise and strong theoretical guarantees. Among them, atomic norm minimization based methods are extremely popular. In~\cite{tang2013compressed}, the authors propose an atomic norm minimization based scheme to super-resolve unknown frequencies of a signal from its random time samples. Yang et al.~\cite{yang2014exact} extend the super-resolution of unknown frequencies in~\cite{tang2013compressed} to the case where multiple measurement vectors are available. Our earlier work~\cite{li2017atomicIEEE} also brings atomic norm minimization into the application of modal analysis for super-resolution of unknown modal parameters of a vibration system from its random and compressed measurements. Meanwhile, the robustness of atomic norm minimization when given noisy data~\cite{candes2013super} has also been widely studied in the past few years. The authors in~\cite{bhaskar2013atomic} apply an atomic norm denoising based technique to line spectral estimation, which is one of the fundamental problems in statistical signal processing. The same authors~\cite{tang2015near} also establish a nearly optimal algorithm, which is called atomic norm soft thresholding, to denoise a mixture of complex sinusoids. In~\cite{li2018approximate}, the authors use atomic norm denoising to investigate the performance of super-resolution line spectral estimation with white noise and provide theoretical guarantees for support recovery. In addition, atomic norm denoising is also studied in~\cite{Li15} for the multiple measurement vector case.

Our work is most closely related to papers~\cite{tang2015near} and~\cite{yang2016super}. The authors in~\cite{tang2015near} focus on a mixture of complex exponentials, while we work on a superposition of complex exponentials with unknown waveform modulations. It can be seen from Section~\ref{prob} that our problem reduces to the problem studied in~\cite{tang2015near} by setting the subspace dimension $K=1$ and selecting the subspace matrix $\B$ as an $N\times 1$ vector with all ones. Both~\cite{tang2015near} and our work establish theoretical guarantees for the mean square error (MSE) with respect to the noise level and true signal parameters. However, since we deal with a more sophisticated scenario in this work, our theory also depends on properties of the subspace matrix $\B$. In addition, we provide an explicit success probability (unlike~\cite{tang2015near}) that increases with the number of signal samples. In~\cite{yang2016super}, the authors study the problem of non-stationary blind super-resolution in a noiseless setting, namely, recovering parameters of a sum of unknown complex exponentials from modulations with unknown waveforms. In contrast, we consider a more practical scenario in which the observed data are contaminated with noise. Therefore, we use different algorithms in this work. In~\cite{yang2016super}, one can exactly recover the unknown parameters with high probability when provided with enough samples. However, it is no longer possible to achieve exact recovery in this work since we only have access to the noisy observations. Thus, the goal of this work is to characterize the MSE as a function of the noise variance and the true signal parameters.
%\SL{(Also we use different algorithms.)}

The main contribution of this work is that we have quantified how the MSE depends on the noise level and the true signal parameters. Namely, we provide a theoretical result to bound the MSE in terms of the noise variance, the total number of uniform samples, the number of true frequencies and the dimension of the subspace in which the unknown waveform modulations live. To be more precise, we have proved both theoretically and numerically that 1) the MSE scales linearly with the noise variance and the subspace dimension, and 2) the MSE is inversely proportional to the total number of uniform samples.
%{\bf [Mike: actually, our upper bound scales like this
%(the true MSE could be better, for example, linear in $J$, although it turns out it does seem to scale like $J^2$).]}
We have proved theoretically that the MSE scales at worst with square of the number of true frequencies but numerical experiments show that it scales linearly with the number of frequencies. We leave the problem of improving our theoretical bound to match the numerical experiments for our future work.

The remainder of this paper is organized as follows. In Section~\ref{prob}, we set up our problem, propose an atomic norm denoising program and introduce its semidefinite program (SDP). In Section~\ref{main}, we present the main theorem that provides the  theoretical guarantee for the atomic norm denoising program. In Section~\ref{nume}, we illustrate the theoretical guarantee with several numerical simulations. The proof for the main theory is presented in Section~\ref{proof_THM_MSE}. Finally, we conclude this work and discuss future direction in Section~\ref{conc}. The Appendix provides some supplementary theoretical results.

\section{Problem Formulation}
\label{prob}

In this work, we consider the following signal
\begin{align}
\xs(m) = \sumj c_j e^{-i2\pi m \tau_j} \g_j(m),~m=-2M,\ldots,0,\ldots,2M,
\label{signal_noiseless}
\end{align}
which can be interpreted as uniform samples of a \textcolor{black}{continuous-time superposition} of complex exponentials with unknown amplitudes and frequencies, each modulated by a different waveform. The requirement for the above sampling indices to be centered around zero is just for technical convenience. Denote $N=4M+1$ as the length of our samples. All the conclusions in this work remain true with appropriate modifications for any $N$ consecutive samples. Without loss of generality, we assume that the unknown coefficients $c_j >0$ and the unknown frequencies $\tau_j$ are normalized, i.e., $\tau_j\in[0,1)$. $\g_j\in\CCC^N$ are unknown waveforms and $J$ is the number of active frequencies in the signal $\xs$.

\textcolor{black}{The signal model introduced in~\eqref{signal_noiseless} appears in a wide range of applications.
%\footnote{\textcolor{black}{We only list a few examples here. The readers can refer to~\cite{yang2016super} for more detailed discussions on some other applications.}}
For example, the authors in~\cite{heckel2016super} consider a radar imaging problem of identifying the relative distances and velocities of targets from a received signal $x(t) = \sumj c_j e^{i2\pi \tau_j t} g(t-\nu_j)$, which can be viewed as a sum of finitely many delays (by $\nu_j$) and Doppler shifts (by $\tau_j$) of a given transmitted signal $g(t)$. The signal model~\eqref{signal_noiseless} can be obtained by sampling the received signal $x(t)$ and defining $\g_j$ to be samples of the $j$th delayed copy of $g(t)$. In addition, the signal in NMR spectroscopy is modeled as $x(t) = \sumj c_j e^{i2\pi \tau_j t} e^{-\nu_j t}$ in~\cite{cai2016robust}. One can again obtain the signal model~\eqref{signal_noiseless} by sampling the NMR spectroscopy signal $x(t)$. Finally, the received signal in multi-user communication systems~\cite{luo2006low} can be modeled as $x(t) = \sum_{j=1}^J c_j g_j(t-\tau_j)$, with each $c_j$ being an unknown coefficient and each $\tau_j$ being an unknown delay of an unknown transmitted signal $g_j(t)$. In this case, the signal model~\eqref{signal_noiseless} can be obtained by sampling the Fourier transform of $x(t)$, namely, $\xh(f) = \sumj c_j e^{-i2\pi f \tau_j}\widehat{g}_j(f)$.}

As noise is ubiquitous in practice, we may only have access to the noisy observations of $\xs$, namely,
\begin{align*}
\y(m) = \xs(m) + \z(m),~m=-2M,\ldots,0,\ldots,2M,
\end{align*}
where $\z$ is the observation noise with i.i.d.\ complex Gaussian entries from the distribution $\CN(0,\sigma^2)$.
To recover the unknown frequencies $\{\tau_j\}$, coefficients $\{c_j\}$ and modulation waveforms $\{\g_j\in\CCC^N\}$ from the noisy observation $\y$, we observe that the number of degrees of freedom ($\OO(JN)$) is much larger than the number of observations ($\OO(N)$), which implies that we need some other assumptions to make the inverse problem well-posed.
%{\bf Mike: We should argue somewhere that this assumption is reasonable in practice, not just that we must make this assumption so that the problem is not ill-posed.}
Therefore, we assume that all the unknown waveforms $\{\g_j\}$ belong to a common and known low-dimensional subspace, which is spanned by the columns of a matrix $\B\in\CCC^{N\times K}$ with $K\leq N$. Let $\b_m\in\CCC^K$ denote the $m$-th column of $\B^H$, i.e.,
\begin{align*}
\B = [\b_{-2M}~\cdots~\b_0~\cdots~\b_{2M}]^H.
\end{align*}
Then, we have $\g_j = \B \h_j$ for some unknown coefficients $\h_j\in\CCC^K$. Without loss of generality, we also assume that $\h_j$ has unit norm, i.e., $\|\h_j\|_2=1$. This is because the coefficients $c_j$ can be scaled as needed. Note that $\g_j$ can be estimated once $\h_j$ is recovered. Therefore, the number of degrees of freedom becomes $\OO(JK)$, which can be smaller than $\OO(N)$ when we have enough measurements, that is, when $N$ is large enough.

\SL{As is illustrated in~\cite{yang2016super}, the assumption that the unknown waveforms $\{\g_j\}$ belong to a common and known low-dimensional subspace appears in many real applications such as super-resolution imaging and multi-user communication systems. For example, the point spread functions in super-resolution imaging can be modeled as Gaussian kernels with unknown widths~\cite{quirin2012optimal,huang2008three}. One can construct a dictionary of Gaussian functions having different widths and apply principal component analysis to this dictionary to discover a low-dimensional subspace that accurately represents the unknown point spread functions. This is also demonstrated in the simulations of~\cite{yang2016super}.
% In addition, it is also reasonable to assume that the unknown waveforms transmitted by different users in a multi-use communication system belong to a subspace~\cite{luo2006low}.
}

For any $\tau\in[0,1)$, define a vector $\a(\tau)\in\CCC^{N}$ as
\begin{align}
\a(\tau)=\left[
e^{i 2\pi \tau (-2M)}~\cdots~1~\cdots~e^{i 2\pi \tau (2M)}\right]^\top.
\label{atom_a}
\end{align}
Then, with the assumption that $\g_j = \B \h_j$, we can rewrite each sample of the signal $\xs$ in~\eqref{signal_noiseless} as
\begin{align*}
\xs(m) &= \sumj c_j \a(\tau_j)^H \e_m \b_m^H \h_j \\
& = \sumj c_j \tr\left( \e_m \b_m^H \h_j \a(\tau_j)^H  \right)\\
& = \tr\left( \e_m \b_m^H \sumj c_j \h_j \a(\tau_j)^H  \right)\\
& = \left\langle \sumj c_j \h_j \a(\tau_j)^H, \b_m \e_m^H    \right\rangle,
\end{align*}
where $\e_m$ is the $(m+2M+1)$-th column of the $N\times N$ identity matrix $\I_N$\footnote{Note that we use $\I$ with a subscript to denote an identify matrix with appropriate size in this work.} and $\tr(\cdot)$ denotes the trace of a matrix. The inner product between two matrices $\X_1$ and $\X_2$ is defined as
\begin{align*}
\langle \X_1,\X_2 \rangle = \tr\left(\X_2^H \X_1\right).
\end{align*}
Define a linear operator $\BB: \CCC^{K\times N} \rightarrow \CCC^N$ and its corresponding adjoint operator $\BB^*: \CCC^N \rightarrow \CCC^{K\times N}$ as
\begin{align}
\left[ \BB(\X) \right]_m &\triangleq \left\langle \X,\b_m \e_m^H  \right\rangle,\label{def_B}\\
\BB^*(\x) &\triangleq \summ  \x(m) \b_m \e_m^H, \label{def_B*}
\end{align}
where $\left[ \BB(\X) \right]_m$ denotes the $m$-th entry of $\BB(\X)$.
Define
\begin{align}
\Xs \triangleq \sumj c_j \h_j \a(\tau_j)^H.
\label{data_noiseless}
\end{align}
Then, we have
\begin{align*}
\xs = \BB(\Xs)
\end{align*}
and the noisy observation vector $\y$ can be rewritten as
\begin{align*}
\y = \xs+\z = \BB(\Xs)+\z.
\end{align*}

When $J$ is small, the noiseless data matrix $\Xs$, defined in~\eqref{data_noiseless}, can be viewed as a sparse combination of elements from the following atomic set
\begin{align*}
\AA \triangleq \left\{\h \a(\tau)^H: \tau\in[0,1), \|\h\|_2=1, \h\in\CCC^K\right\}
\end{align*}
with $\a(\tau)$ defined in~\eqref{atom_a}.
The associated atomic norm is then defined as
\begin{equation}
\begin{aligned}
\|\X\|_{\AA} &\triangleq \inf \{t>0: \X\in t \change{conv}(\AA) \}\\
&= \inf_{c_j,\tau_j,\|\h_j\|_2=1} \left\{ \sumj c_j: \X = \sumj c_j \h_j \a(\tau_j)^H, c_j>0   \right\},
\label{atomic_norm}
\end{aligned}
\end{equation}
where $\change{conv}(\AA)$ is the convex hall of the atomic set $\AA$.

Let $\lambda$ denote a suitably chosen regularization parameter.\footnote{Please refer to Section~\ref{bound_regularization} for guidelines on choosing $\lambda$.}
Inspired by the sparse representation of $\Xs$ with respect to the atomic set $\AA$, we perform denoising by proposing an atomic norm regularized least-squares problem
\begin{align}
\Xh = \arg\min_{\X} \frac 1 2 \|\y-\BB(\X)\|_2^2 + \lambda\|\X\|_{\AA},
\label{atomic_denoising}
\end{align}
which is equivalent to the following semidefinite program (SDP)\footnote{
In practice, one can use the CVX software package to solve this SDP~\cite{grant2008cvx}.}~\cite{yang2014exact,Li15}
\begin{align*}
\{ \Xh, \uh, \Th \} = \arg\min_{\X,\u,\T}&~ \frac 1 2 \|\y-\BB(\X)\|_2^2 + \frac{\lambda}{2}\left(\frac{1}{N} \tr(\change{Toep}(\u))+  \tr(\T)   \right)\\
\st&\left[ \begin{array}{cc}
\change{Toep}(\u)& \X^H\\
\X&\T
\end{array}
\right]\succeq 0.
\end{align*}
Here, $\change{Toep}(\u)\in\CCC^{N\times N}$ denotes the Hermitian Toeplitz matrix with $\u$ being its first column.

In this work, our goal is to analyze the performance of the above denoising scheme by bounding the mean-squared error (MSE) $\frac{1}{N}\|\xh-\xs\|_2^2$ between the solution $\xh = \BB(\Xh)$ and the true signal $\xs = \BB(\Xs)$.
%We use $\Xh$ to denote solution of the atomic norm denoising problem~\eqref{atomic_denoising}.

\section{Theoretical Guarantee for Atomic Norm Denoising}
\label{main}

Motived by~\cite{candes2011probabilistic, chi2016guaranteed, yang2016super}, we assume that the columns of $\B^H$, i.e., $\b_m\in\CCC^K$, are independent and identically distributed (i.i.d.) samples from a distribution that satisfies the isotropy and incoherence properties with coherence parameter $\mu$.
%\footnote{\SL{As is evidenced by the numerical experiments in~\cite{yang2016super}, the randomness assumption on $\B$ does not appear to be critical in pratice.}}
\begin{itemize}
\item \textbf{Isotropy property:} A distribution $\FF$ satisfies the isotropy property if\footnote{Note that this definition of isotropy property is slightly different from the one used in~\cite{candes2011probabilistic, chi2016guaranteed, yang2016super}. To give an example of $\b$ that obeys the isotropy and incoherence properties~\eqref{isotprop} and~\eqref{incoprop} with $\mu = 1$, we can choose the entries of $\b$ from the Rademacher distribution. }
\begin{align}
\EEE \b \b^H = \I_K,~\text{and}~\EEE\left( \frac{\b}{\|\b\|_2} \frac{\b^H}{\|\b\|_2}\right) = \frac 1 K \I_K,~~ \b\in\FF.
\label{isotprop}
\end{align}
\item \textbf{Incoherence property:} A distribution $\FF$ satisfies the incoherence property with coherence $\mu$ if
\begin{align}
\max_{1\leq k \leq K} |\b(k)|^2 \leq \mu,~~ \b\in\FF
\label{incoprop}
\end{align}
holds almost surely. Here, $\b(k)$ denotes the $k$-th entry of $\b$.
\end{itemize}

Next, we present the main result that characterizes the MSE $\frac{1}{N}\|\xh-\xs\|_2^2$ in the following theorem.

\begin{Theorem}
\label{THM_MSE}
Suppose the noiseless signal $\xs$ is given as in~\eqref{signal_noiseless} with the true frequencies satisfying a minimum separation condition
\begin{align}
\min_{j\neq k} d(\tau_j,\tau_k) \geq \frac 1 N,
\label{minsepcon}
\end{align}
where $d(\tau_j,\tau_k) \triangleq |\tau_j-\tau_k|$ denotes the wrap-around distance on the unit circle. Assume that we have $N=4M+1$ noisy measurements $\y(m) = \xs(m) + \z(m),~m=-2M,\ldots,0,\ldots,2M$, where $\z(m)$ is i.i.d.\ complex Gaussian noise with mean 0 and variance $\sigma^2$. We also assume that $\b_m$ are i.i.d.\ samples from a distribution that satisfies the isotropy~\eqref{isotprop} and incoherence~\eqref{incoprop} properties with coherence parameter $\mu$. Then, the estimate $\xh$ obtained by solving the atomic norm regularized least-squares problem~\eqref{atomic_denoising} with $\lambda = 2 \eta \sigma \|\B\|_F \sqrt{\log(N)}$ for some $\eta > 1$ satisfies
\begin{align}
\frac 1 N \|\xh-\xs\|_2^2 \leq C \eta^2 \sigma^2 \mu^2   \frac{J^2K}{N} \log(N) \log\left(NJK\right)
\label{error_prob}
\end{align}
with probability at least $1-cN^{-1}$ if $\eta \!\in\! (1 ,\infty )$ is chosen sufficiently large and $N \geq C \mu J^2 K \log\left(NJK \right)$.
Here, $C$ and $c$ are some numerical constants.\footnote{Note that the numerical constants $C$ and $c$ used in this paper can vary from line to line.}
%\note{\SL{Plugging $\|\B\|_F^2 \leq \mu KN$?}}
\end{Theorem}

\SL{Our use of the isotropy~\eqref{isotprop} and incoherence~\eqref{incoprop} properties parallels the assumptions made for subspace models in several related works. These properties were first defined in~\cite{candes2011probabilistic} for the development of a probabilistic and RIPless theory of compressed sensing, and then used in~\cite{chi2016guaranteed} for the problem of blind spikes deconvolution. Other random subspace assumptions are also used in~\cite{asif2009random, ahmed2013blind} for random channel coding and blind deconvolution. The transmitted signal in multi-user communication systems~\cite{luo2006low} can also be represented in a known low-dimensional random subspace in the case when each of the transmitters sends out a random signal for the sake of security, privacy, or spread spectrum communications.\footnote{\textcolor{black}{Random signals also appear in applications such as noise radar~\cite{axelsson2004noise}. The transmitted random signals can be either directly generated from a noise-generating microwave source or obtained by modulating a sine wave with random noise.}} Note that the matrix $\B$ with randomness assumptions can alternatively be viewed as a sensing measure to obtain random measurements of the data matrix $\Xs$ via~\eqref{def_B}. As stated in~\cite{candes2011probabilistic} and many other compressive sensing works, random measurements are crucial in the development of theoretical results, and can result in better empirical results. As evidenced by the numerical experiments in~\cite{yang2016super}, the randomness assumption on $\B$ does not appear to be critical in practice.} Finally, observe that we have used not only the same isotropy and incoherence properties on the random subspace model as in papers~\cite{chi2016guaranteed, yang2016super}, but also some other randomness assumptions on the subspace, namely,
$\EEE\left( \frac{\b}{\|\b\|_2} \frac{\b^H}{\|\b\|_2}\right) = \frac 1 K \I_K$, where $\b$ denotes arbitrary column of $\B^H$.
\SL{This part of the isotropy property is only an artifact of our proof technique.}
%the normalized vector $\frac{\b_m}{\|\b_m\|_2}$ are uniform on the unit sphere.

%\SL{[Provide a lower bound to evaluate how sharp~\eqref{error_prob} is.]}

\SL{Note that an oracle MSE rate of $\OO(\sigma^2 KJ/N)$ could be achieved if one had enough measurements (i.e., $N\geq KJ$) and perfect knowledge of the well-separated frequencies~$\tau_j$. To be more precise, the noiseless signal $\xs$ in~\eqref{signal_noiseless} can be written as $\xs = \M_{ab} \v_{ch}^\star$, where $\M_{ab}\in\CCC^{N\times KJ}$ is a matrix related to the subspace matrix $\B$ and vectors $\a(\tau_j)$, and $\v_{ch}^\star \in \CCC^{KJ}$ is a vector determined by the coefficients $c_j$ and $\h_j$. If $N\geq KJ$ and the well-separated frequencies~$\tau_j$ are known, the matrix $\M_{ab}$ is then known and is a tall matrix with full column rank due to the randomness of $\B$. Then, recovering $\xs$ from its noisy observation $\y = \xs + \z$ is equivalent to solving a least-squares problem. Elementary calculations then show that $\frac 1 N \EEE_{\z}\|\xh-\xs\|_2^2 = \sigma^2 \frac{KJ}{N}$. Therefore, our proposed MSE bound~\eqref{error_prob} is optimal except for an extra $J$ factor and some logarithmic terms. This extra $J$ factor may be removed by using some other proof strategies instead of the dual analysis of atomic norm. We leave the problem of improving our theoretical bound for future work.
}

Finally, note that we have removed the randomness assumption on $\h$ as is required in~\cite{yang2016super}. However, we bound the sample complexity $N$ with $\OO(J^2K)$ instead of $\OO(JK)$. We also notice that the author in~\cite{chi2016guaranteed} provides a sample complexity bound $\OO(J^2K^2)$ without using any randomness assumptions on $\h$.
%Moreover, due to the fact that $\EEE(\|\B\|_F^2) = NK$, we can conclude that the MSE $\frac 1 N \|\xh-\xs\|_2^2$ is on the order of $\OO\left(\sigma^2 \frac{J^2K}{N}\log(N) \log\left(\frac{K(J+1)}{\delta}\right)\right)$.
It is worth noting that those papers are not focused on signal denoising as we are in this work.
%\SL{[Provide more discussion. e.g, why $J^2$ vs. $J$ arises in this analysis.]}
\SL{The extra $J$ factor in our sample complexity bound is a result of removing the randomness assumption on $\h$ that is used in Lemmas~11 and 13 of reference~\cite{yang2016super}. %When compared with the above lower bound $\OO(\sigma^2 KJ/N)$ obtained in the case when we exactly know the well-separated frequencies, our MSE bound given in~\eqref{error_prob} may also contain an extra $J$ factor. 
}

\section{Numerical Simulations}
\label{nume}
%\note{\SL{ Needs more work!}}

In this section, we conduct four numerical experiments to support our theoretical bound in Theorem~\ref{THM_MSE}. In all of these experiments, we perform denoising by solving the SDP of the atomic norm regularized least-squares problem~\eqref{atomic_denoising} with CVX. As is suggested in Theorem~\ref{THM_MSE}, we set the regularization parameter as $\lambda = 2 \eta \sigma \|\B\|_F \sqrt{\log(N)}$ with $\eta = 0.5$\footnote{Note that we require $\eta > 1$ in Theorem~\ref{THM_MSE}. However, we find that $\eta = 0.5$ can achieve much lower MSE in practice. Thus, we set $\eta = 0.5$ for all of the following experiments.}. We generate the entries of $\B$ as Rademacher random variables and the entries of $\h_j,~j=1,\ldots,J$ as Gaussian random variables satisfying $\NN(0,1)$ and then normalize all the $\h_j$ to make sure $\|\h_j\|_2 = 1$. We generate the observation noise $\z$ with i.i.d.\ complex Gaussian entries satisfying $\CN(0,\sigma^2)$. 50 trials are performed for each of these experiments.

In the first experiment, we examine the relationship between the MSE $\frac 1 N \|\xh-\xs\|_2^2$ and the total number of uniform samples $N$ with $J = 3,~K=4,$ and $\sigma=0.1$. The true frequencies and corresponding amplitudes are set as $\tau_1 = 0.1,~\tau_2=0.15,~\tau_3=0.5$ and $c_1 = 1,~c_2 = 2,~c_3 = 3$. We change $M$ from 10 to 100, namely, $N = 4M+1$ changes from 41 to 401.
Figure~\ref{test_MSE_y}(a) shows the denoising performance of atomic norm regularized least-squares problem~\eqref{atomic_denoising} while Figure~\ref{test_MSE_y}(b) indicates that the MSE does scale with $\OO(\frac{1}{N}\log(N))$, which implies that the MSE can decrease linearly (if we ignore the log term) as we increase the number of uniform samples $N$.

In the second experiment, we characterize the impact of the noise variance $\sigma^2$ on MSE. We fix $M=20$ and set $\sigma = 0.025:0.025:0.25$. Other parameters are set the same as the first experiment. It can be seen from Figure~\ref{test_MSE}(a) that the MSE does scale linearly with $\sigma^2$, as is shown in Theorem~\ref{THM_MSE}.

To see the influence of the number of true frequencies $J$ on the MSE, we repeat the first experiment by fixing $M=20$ and changing $J$ from 1 to 10. We randomly select $J$ true frequencies from a set $\{0:\change{MinSep}:1-\change{MinSep}\}$ with $\change{MinSep}=\frac{1}{M}$ denoting the minimum separation. The corresponding amplitudes are then set as $c_j = j,~j=1,\ldots,J$. Other parameters are the same as those used in the first experiment.
Figure~\ref{test_MSE}(b) implies that the MSE scales with $J\log(J)$, which is better than the one indicated in Theorem~\ref{THM_MSE}. We leave the improvement of Theorem~\ref{THM_MSE} for future work.

Finally, we illustrate the relationship between the MSE and the subspace dimension $K$ in the last experiment. We set $M=20$ and change $K$ from 1 to 10. Other parameters are same as those used in the first experiment.
Figure~\ref{test_MSE}(c) roughly shows a linear relationship between the MSE and $K\log(K)$, which is consistent with the bound in Theorem~\ref{THM_MSE}.

\begin{figure}[t]
\begin{minipage}{0.49\linewidth}
\centering
\includegraphics[width=2.7in]{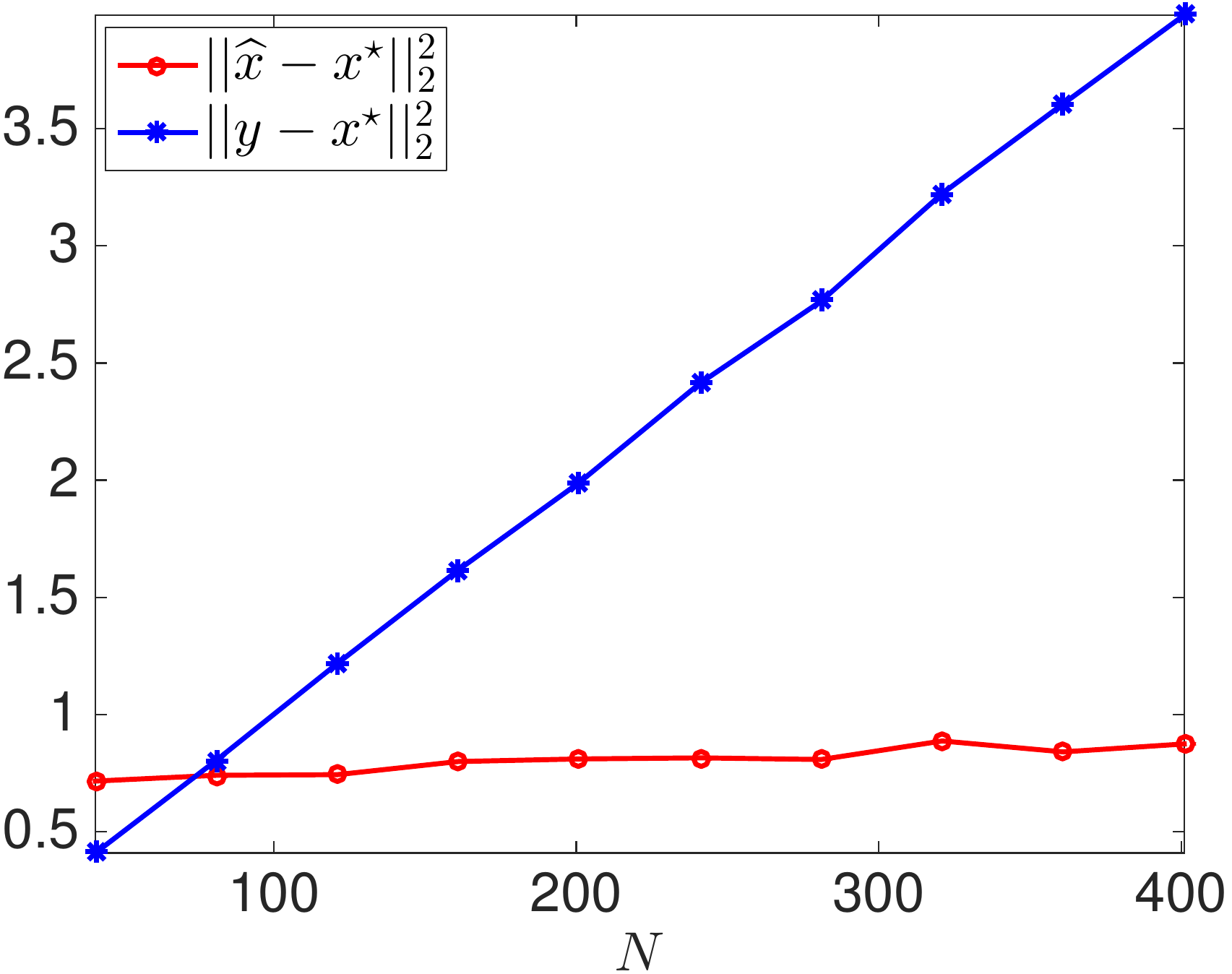}
\centerline{\small{(a) $J=3,~K=4,~\sigma=0.1$}}
\end{minipage}
\hfill
\begin{minipage}{0.49\linewidth}
\centering
\includegraphics[width=2.9in]{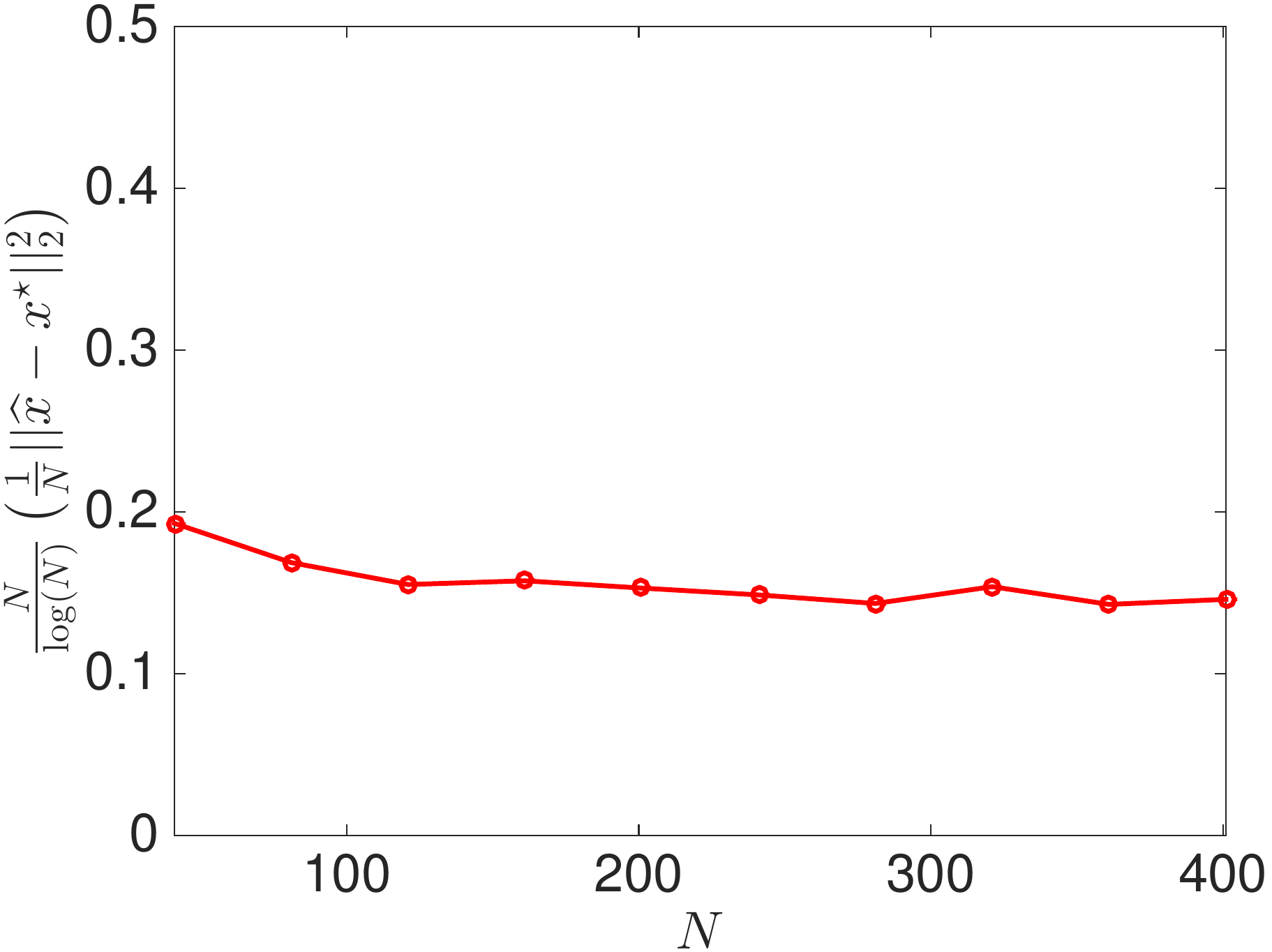}
\centerline{\small{(b) $J=3,~K=4,~\sigma=0.1$}}
\end{minipage}
\caption{(a) the denoising performance of atomic norm regularized least-squares problem~\eqref{atomic_denoising}, (b) the relationship between the scaled MSE $\frac{N}{\log(N)}\left(  \frac 1 N \|\xh-\xs\|_2^2\right)$ and the total number of uniform samples $N$. }
\label{test_MSE_y}
\end{figure}

\begin{figure}[t]

\centering
\includegraphics[width=2.9in]{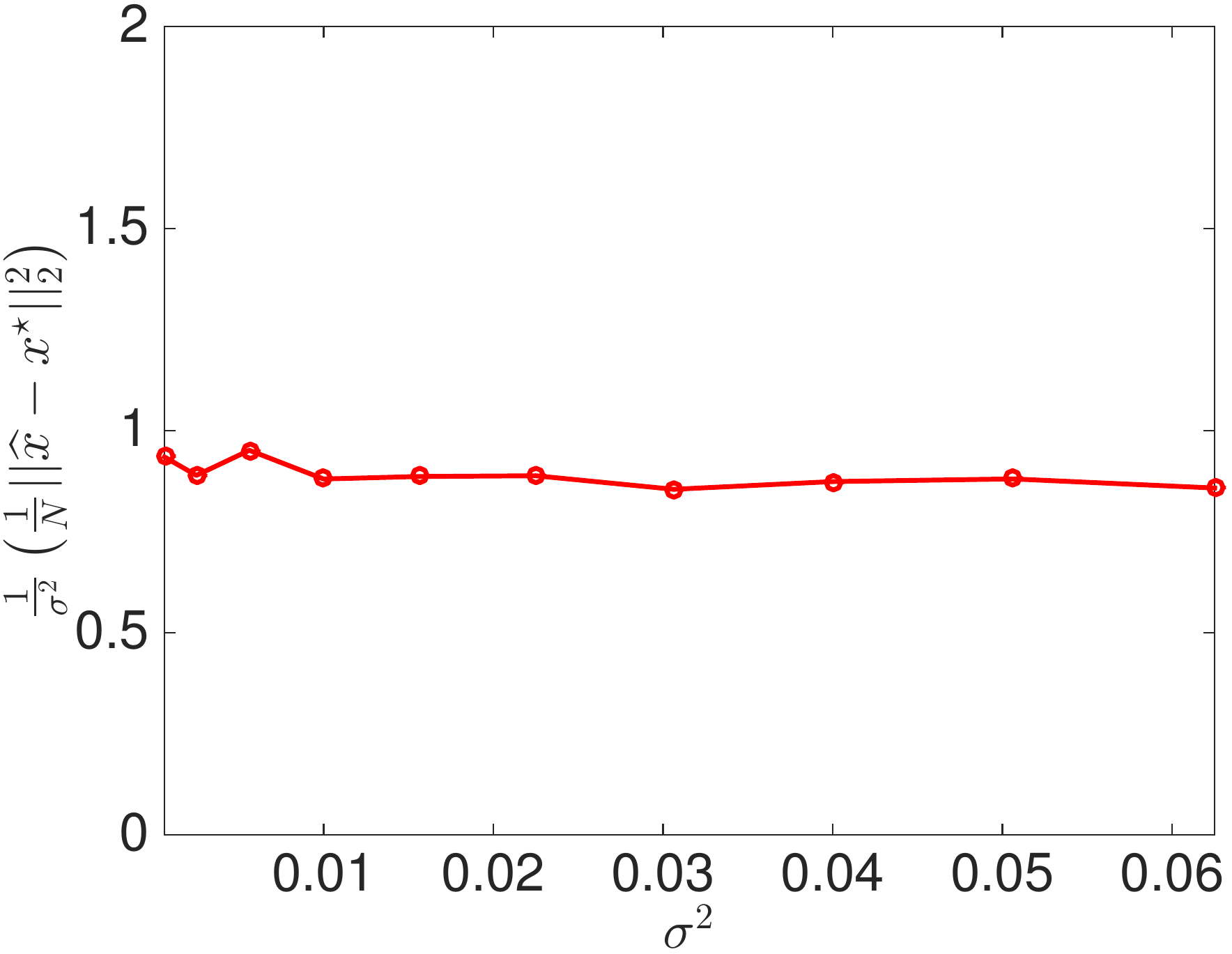}
\centerline{\small{(a) $J=3,~K=4,~M=20$}}
\\
\begin{minipage}{0.49\linewidth}
\centering
\includegraphics[width=2.9in]{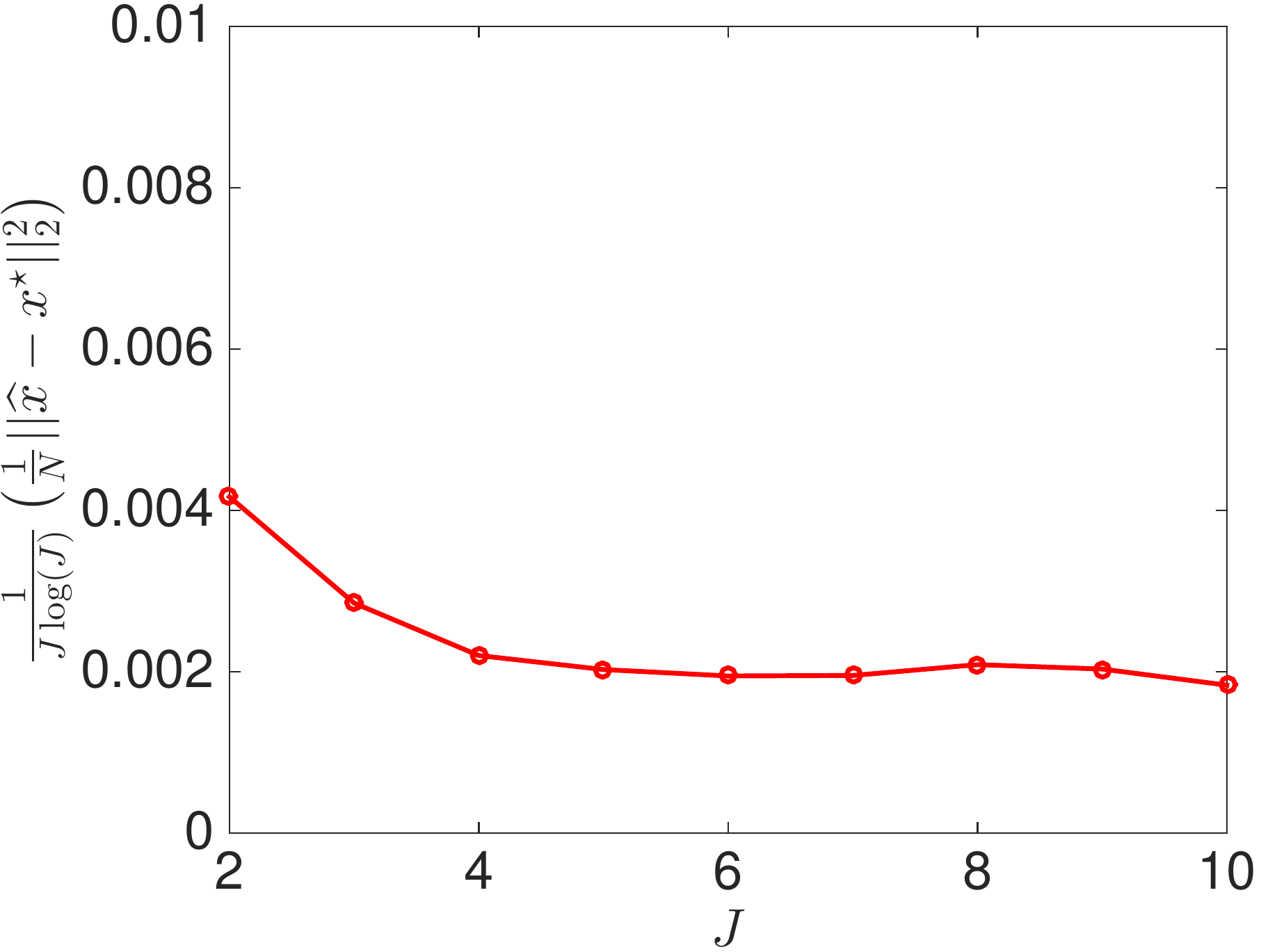}
\centerline{\small{(b) $M=20,~K=4,~\sigma=0.1$}}
\end{minipage}
\hfill
\begin{minipage}{0.49\linewidth}
\centering
\includegraphics[width=2.9in]{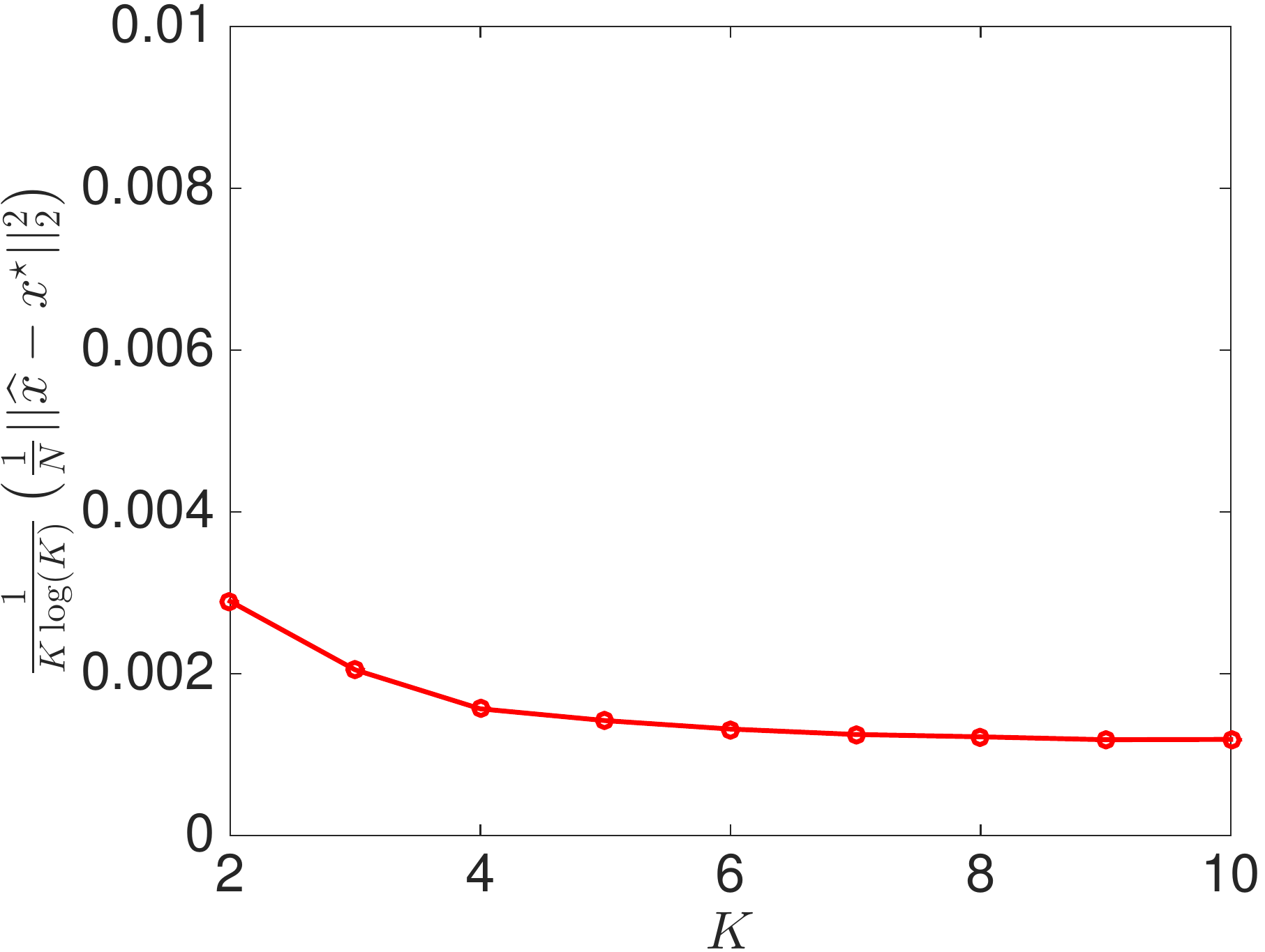}
\centerline{\small{(c) $M=20,~J=3,~\sigma=0.1$}}
\end{minipage}
\caption{The relationship between (a) the scaled MSE $ \frac{1}{\sigma^2} \left( \frac 1 N \|\xh-\xs\|_2^2\right)$ and noise variance $\sigma^2$, (b) the scaled MSE $ \frac{1}{J\log(J)} \left( \frac 1 N \|\xh-\xs\|_2^2\right)$ and the number of true frequencies $J$, (c) the scaled MSE $ \frac{1}{K\log(K)} \left( \frac 1 N \|\xh-\xs\|_2^2\right)$ and the subspace dimension $K$. }
\label{test_MSE}
\end{figure}

%{\bf [Mike: It seems like our decay rate is slower than 0.4/N. maybe we need to look at this more carefully?]  }

\section{Proof of Theorem~\ref{THM_MSE}}
\label{proof_THM_MSE}

%\SL{[Add an overview of the proof here.]}

\SL{The proof of Theorem~\ref{THM_MSE} is presented in two separate subsections. In Subsection~\ref{bound_regularization}, we discuss how to choose the regularization parameter $\lambda$.}
It is well known that a good choice of the regularization parameter $\lambda$ can achieve accelerated convergence rates.\footnote{\SL{The conditions under which an accelerated convergence rate can be obtained for an atomic norm denoising problem is discussed in~\cite{bhaskar2013atomic}.}} \SL{In Subsection~\ref{bound_error_pr}, with the well chosen regularization parameter $\lambda$, we bound the MSE with respect to the noise level and true signal parameters.} We prove Theorem~\ref{THM_MSE} by extending the proof of~\cite[Theorem 1]{tang2015near} and~\cite[Theorem III.6]{li2017atomicIEEE} to our framework. \SL{Note that~\cite[Theorem III.6]{li2017atomicIEEE} is a multiple measurement vector (MMV) extension of~\cite[Theorem 1]{tang2015near}. Due to the random linear operator $\BB$ that appears in our atomic norm denoising problem~\eqref{atomic_denoising}, we will develop a random extension of the two previous results, \cite[Theorem 1]{tang2015near} and~\cite[Theorem III.6]{li2017atomicIEEE}, with respect to $\BB$. The proof idea is borrowed from these prior works. However, the proof there does not extend directly to our framework with the linear operator $\BB$. For example, this linear operator $\BB$ can first affect our choice of the regularization parameter $\lambda$.
}

Inspired by the  regularization parameter used in prior works~\cite{bhaskar2013atomic,tang2015near, li2017atomicIEEE}, we set
\begin{align*}
%$
\lambda \approx \eta \EEE_{\z}\|\BB^*(\z)\|_{\AA}^*
%$
\end{align*}
in the atomic norm regularized least-squares problem~\eqref{atomic_denoising}. Here, $\z$ is a complex Gaussian vector and $\eta\in(1,\infty)$ is a constant that must be large enough to enable the proof of Lemma~\ref{FI2}.
$\|\cdot\|_{\AA}^*$ is defined as the dual norm of the atomic norm defined in~\eqref{atomic_norm}, which is given as
\begin{equation}
\begin{aligned}
\|\Q\|_{\AA}^* &\triangleq \sup_{\|\X\|_{\AA}\leq 1} \langle \Q,\X \rangle_{\RRR}\\
&=\sup_{\tau\in[0,1),\|\h\|_2=1} \left\langle  \Q,\h\a(\tau)^H  \right\rangle_{\RRR}\\
&=\sup_{\tau\in[0,1),\|\h\|_2=1} \left\langle  \Q \a(\tau),\h  \right\rangle_{\RRR}\\
&=\sup_{\tau\in[0,1)} \| \Q \a(\tau)\|_2.
\label{dual_atomic_norm}
\end{aligned}
\end{equation}
%\note{\SL{I moved the definition of dual norm here from the proof of Lemma 5.1.}}
In order to set $\lambda$, we first compute an upper bound for $\EEE_{\z}\|\BB^*(\z)\|_{\AA}^*$.

\subsection{Bounding \texorpdfstring{$\EEE_{\z}\|\BB^*(\bm{z})\|_{\AA}^*$}{$\EEE_z\|\BB^*(z)\|_{\AA}^*$} and \texorpdfstring{$\PPP\left(\|\BB^*(\bm{z})\|_{\AA}^* \leq \frac{\lambda}{\eta} \right)$}{$\PPP_z\left(\|\BB^*(z)\|_{\AA}^* \leq \frac{\lambda}{\eta} \right)$}}
\label{bound_regularization}

\begin{Lemma} \label{Lemma_EPBz}
Let $\z\in\CCC^N$ be a random vector with i.i.d.\ complex Gaussian entries from the distribution $\CN(0,\sigma^2)$. Define a linear operator $\BB^*: \CCC^N \rightarrow \CCC^{K\times N}$ as in~\eqref{def_B*}, i.e.,
\begin{align*}
\BB^*(\x) &= \summ  \x(m) \b_m \e_m^H,
\end{align*}
where $\b_m$ is the $m$-th column of a $K\times N$ matrix $\B^H$, the Hermitian matrix of $\B$. $\x(m)$ is the $m$th entry of $\x$ and $\e_m$ is the $(m+2M+1)$-th column of the $N\times N$ identity matrix $\I_N$.
%Assume that the entries of $\b_m$ satisfy the standard normal distribution for $m=-2M,\ldots,2M$.
Then, there exists a numerical constant $C\in(1,2)$ such that
\begin{align}
\EEE_{\z}\|\BB^*(\z)\|_{\AA}^* \leq C \sigma \|\B\|_F \sqrt{\log(N)}. \label{bound_EBz}
\end{align}
By setting the regularization parameter as $\lambda = 2\eta \sigma \|\B\|_F \sqrt{\log(N)}$, we have
\begin{align}
\PPP\left( \|\BB^*(\z)\|_{\AA}^* \leq \frac{\lambda}{\eta} \right)
\geq 1- c\frac{1}{N}
\label{bound_PBz}
\end{align}
with $c$ being some constant.
\end{Lemma}

\begin{proof}

The dual norm defined in~\eqref{dual_atomic_norm} implies that
\begin{align*}
\left(\|\BB^*(\z)\|_{\AA}^*\right)^2 &= \sup_{\tau\in[0,1)} \| \BB^*(\z) \a(\tau) \|_2^2\\
&= \sup_{\tau\in[0,1)} \sumk \left| \summ \z(m) \b_m(k) e^{i2\pi \tau m}  \right|^2\\
&= \sup_{\tau\in[0,1)} \sumk \sum_{m,n=-2M}^{2M} \z(m)\z(n)^H\b_m(k)\b_n(k)^H e^{i 2\pi\tau(m-n)}\\
&=\sup_{\tau\in[0,1)} \ZZ_N(e^{i 2\pi \tau}),
\end{align*}
where $\b_m(k)$ is the $k$-th entry of $\b_m$. We have defined a polynomial $\ZZ_N(e^{i 2\pi \tau})$ as
\begin{align*}
\ZZ_N(e^{i 2\pi \tau}) \triangleq \sumk \sum_{m,n=-2M}^{2M} \z(m)\z(n)^H\b_m(k)\b_n(k)^H e^{i 2\pi\tau(m-n)}.
\end{align*}

Note that we have
\begin{align*}
&\ZZ_N(e^{i 2\pi \tau_1})-\ZZ_N(e^{i 2\pi \tau_2})\\
\leq &\left| e^{j2\pi \tau_1}- e^{j2\pi \tau_2} \right| \sup_{\tau\in [0,1)} \ZZ_N'(e^{i 2\pi \tau})\\
=&2|\sin \pi(\tau_1-\tau_2)| \sup_{\tau\in [0,1)} \ZZ_N'(e^{i 2\pi \tau})\\
\leq & 2\pi |\tau_1-\tau_2| \sup_{\tau\in [0,1)} \ZZ_N'(e^{i 2\pi \tau})\\
\leq& 2\pi N|\tau_1-\tau_2| \sup_{\tau\in [0,1)} \ZZ_N(e^{i 2\pi \tau})
\end{align*}
for any $\tau_1, \tau_2\in[0,1)$. The first inequality follows from the mean value theorem while the last inequality follows from Bernstein's inequality for polynomials~\cite{schaeffer1941inequalities}.

Let $\tau_2$ take any of the values $0, \frac 1 L, \ldots, \frac{L-1}{L}$, which gives us
\begin{align*}
\sup_{\tau\in[0,1)}\ZZ_N(e^{i 2\pi \tau}) \leq \max_{l=0,\ldots,L-1} \ZZ_N(e^{i 2\pi l/L} )+\frac{2\pi N}{L} \sup_{\tau\in [0,1)} \ZZ_N\left(e^{i 2\pi \tau} \right).
\end{align*}
Then, we upper bound $\left(\|\BB^*(\z)\|_{\AA}^*\right)^2$ with
\begin{equation}
\begin{aligned}
\left(\|\BB^*(\z)\|_{\AA}^*\right)^2&=\sup_{\tau\in[0,1)} \ZZ_N (e^{i 2\pi \tau} ) \\
&\leq \left(1-\frac{2 \pi N}{L}\right)^{-1}\max_{l=0,\ldots,L-1} \ZZ_N(e^{i 2\pi l/L} )\\
&\leq \left(1+\frac{4 \pi N}{L}\right)\max_{l=0,\ldots,L-1} \ZZ_N\left(e^{i 2\pi l/L} \right)
\label{Bz}
\end{aligned}
\end{equation}
if $L\geq 4\pi N$. It follows that
\begin{align*}
\|\BB^*(\z)\|_{\AA}^* \leq \left(1+\frac{4 \pi N}{L}\right)^{\frac{1}{2}}\left[\max_{l=0,\ldots,L-1} \ZZ_N\left(e^{i 2\pi l/L} \right)\right]^{\frac{1}{2}}
\end{align*}
and
\begin{align}
\EEE_{\z}\|\BB^*(\z)\|_{\AA}^*\leq \left(1+\frac{4 \pi N}{L}\right)^{\frac{1}{2}}\left\{ \EEE_{\z}\left[\max_{l=0,\ldots,L-1} \ZZ_N\left(e^{i 2\pi l/L} \right) \right] \right\}^{\frac{1}{2}}.
\label{EBz}
\end{align}

Observe that, conditioned on $\{\b_m\}$,
\begin{equation}
\begin{aligned}
&\EEE_{\z} \left[\max_{l=0,\ldots,L-1} \ZZ_N\left(e^{i 2\pi l/L} \right)\right] \\
= & \EEE_{\z} \left[\max_{l=0,\ldots,L-1} \sumk \left| \summ \z(m) \b_m(k) e^{i 2\pi (l/L) m }        \right|^2  \right] \\
\leq & \sumk \EEE_{\z} \left[\max_{l=0,\ldots,L-1} \left| \summ \z(m) \b_m(k) e^{i 2\pi (l/L) m }        \right|^2  \right] \\
\triangleq& \frac 1 2 \sigma^2 N \sumk \EEE_{\z}  \left[\max_{l=0,\ldots,L-1} 2 |u_{k,l}|^2  \right],
\label{EZN}
\end{aligned}
\end{equation}
where $u_{k,l}$ is a complex Gaussian random variable and defined as
\begin{align}
u_{k,l} \triangleq \frac{1}{\sigma \sqrt{N}} \summ \z(m) \b_m(k) e^{i 2\pi (l/L)m }
\label{def_ukl}
\end{align}
for $k=1,\ldots,K,~l=0,\ldots,L-1$. Note that the expectation and variance of $u_{k,l}$ are given as
\begin{align*}
\EEE_{\z}(u_{k,l}) &= \frac{1}{\sigma \sqrt{N}} \summ \EEE_{\z}[ \z(m) ]  \b_m(k) e^{i 2\pi (l/L)m } = 0,\\
\change{Var}(u_{k,l}) &= \frac{1}{\sigma^2 N} \summ \change{Var}[ \z(m)] |\b_m(k)|^2 \\
&=\frac{1}{ N} \summ  |\b_m(k)|^2\\
&\triangleq \sigmat_k^2
\end{align*}
since $\z(m)\sim \CN(0,\sigma^2)$.
Therefore, conditioned on $\{\b_m\}$, the complex Gaussian random variable $u_{k,l}$ defined in~\eqref{def_ukl} satisfies $\CN(0,\sigmat_k^2)$. Let $u_{k,l}^r$ and $u_{k,l}^i$ denote the real part and imaginary part of $u_{k,l}$, i.e., $u_{k,l} = u_{k,l}^r+i u_{k,l}^i$. Then, we have
\begin{align*}
u_{k,l}^r \sim \NN(0, \frac 1 2 \sigmat_k^2),~\text{and}~u_{k,l}^i \sim \NN(0, \frac 1 2 \sigmat_k^2),
\end{align*}
which implies that
\begin{align*}
\frac{2|u_{k,l}|^2}{\sigmat_k^2}  = \left(\frac{\sqrt{2}u_{k,l}^r}{\sigmat_k}\right)^2 +\left(\frac{\sqrt{2}u_{k,l}^i}{\sigmat_k}\right)^2
\end{align*}
is a chi-squared random variable with two degrees of freedom since both $\frac{\sqrt{2}u_{k,l}^r}{\sigmat_k}$ and $\frac{\sqrt{2}u_{k,l}^i}{\sigmat_k}$ satisfy standard normal distribution.
Using to the properties of the chi-square distribution, we have
\begin{equation}
\begin{aligned}
&\EEE_{\z}\left[\max_{l=0,\ldots,L-1}  2|u_{k,l}|^2   \right] \\
=& \int_0^{\infty} \PPP\left\{  \max_{l=0,\ldots,L-1}  2|u_{k,l}|^2  \geq t   \right\} dt\\
=& \int_0^{\deltat_k} \PPP\left\{  \max_{l=0,\ldots,L-1}  2|u_{k,l}|^2  \geq t   \right\} dt + \int_{\deltat_k}^{\infty} \PPP\left\{  \max_{l=0,\ldots,L-1}  2|u_{k,l}|^2  \geq t   \right\} dt\\
\leq & \deltat_k + L\int_{\deltat_k}^{\infty} \PPP\left\{ \frac{  2|u_{k,l}|^2}{\sigmat_k^2}  \geq \frac{t}{\sigmat_k^2}  \right\} dt \\
=& \deltat_k + L\int_{\deltat_k}^{\infty} e^{-\frac{t}{2\sigmat_k^2}} dt \\
=& \deltat_k + 2L\sigmat_k^2e^{-\frac{\deltat_k}{2\sigmat_k^2}}.
\label{Eukl}
\end{aligned}
\end{equation}

Choosing $\deltat_k = 2 \sigmat_k^2\log(L)$ and $L = 4\pi N\log(N)$, together with inequalities~\eqref{EBz}, \eqref{EZN}, and~\eqref{Eukl}, we finally obtain
\begin{align*}
\EEE_{\z}\|\BB^*(\z)\|_{\AA}^*
&\leq \left(1+\frac{4 \pi N}{L}\right)^{\frac{1}{2}}\left\{ \EEE_{\z}\left[\max_{l=0,\ldots,L-1} \ZZ_N\left(e^{i 2\pi l/L} \right) \right] \right\}^{\frac{1}{2}}\\
&\leq  \left(1+\frac{4 \pi N}{L}\right)^{\frac{1}{2}} \left\{\frac 1 2 \sigma^2 N \sumk \EEE_{\z}  \left[\max_{l=0,\ldots,L-1} 2 |u_{k,l}|^2  \right] \right\}^{\frac{1}{2}}\\
&\leq  \left(1+\frac{1}{\log(N)}\right)^{\frac{1}{2}} \left\{ \sigma^2 N \sumk [ \sigmat_k^2\log(4\pi N\log(N))+ \sigmat_k^2 ]  \right\}^{\frac{1}{2}}\\
&= \left(1+\frac{1}{\log(N)}\right)^{\frac{1}{2}} \sigma \sqrt{N [ \log(N) + \log(4\pi\log(N))+1]\sumk \sigmat_k^2 }\\
&= \left(1+\frac{1}{\log(N)}\right)^{\frac{1}{2}} \sigma \|\B\|_F \sqrt{ \log(N) + \log(4\pi\log(N))+1 }\\
&\leq C \sigma \|\B\|_F \sqrt{\log(N)},
\end{align*}
where $C$ is a numerical constant that belongs to the interval $(1,2)$ when $N$ is large. Note that the last equality follows from the fact that
\begin{align*}
\sumk \sigmat_k^2 = \sumk \frac{1}{ N} \summ  |\b_m(k)|^2 = \frac 1 N \|\B\|_F^2.
%\label{normB}
\end{align*}
This completes the proof for inequality~\eqref{bound_EBz}.

Next, we can set the regularization parameter $\lambda$ as
\begin{align*}
\lambda =2 \eta \sigma \|\B\|_F \sqrt{\log(N)}
\end{align*}
for some constant $\eta\in(1,\infty)$ and continue to prove inequality~\eqref{bound_PBz}.
It follows from~\eqref{Bz} that
\begin{equation}
\begin{aligned}
\left(\|\BB^*(\z)\|_{\AA}^*\right)^2 &\leq \left(1-\frac{2\pi N}{L} \right)^{-1} \max_{l=0,\ldots,L-1} \ZZ_N\left(e^{i 2\pi l/L} \right)\\
%&= \left(1+4\pi \right) \left[\max_{l=0,\ldots,N-1}  \sumk \left| \summ \z(m) \b_m(k) e^{i 2\pi (l/N) m}  \right|^2 \right] \\
&\triangleq \left(1-\frac{2\pi N}{L} \right)^{-1} \max_{l=0,\ldots,L-1} \sumk \left| \WW_{l,k}  \right|^2,
\label{Bz_temp1}
\end{aligned}
\end{equation}
where $ \WW_{l,k} \triangleq \summ \z(m) \b_m(k) e^{i 2\pi (l/L) m}$ is a set of complex Gaussian variables with mean 0 and variance $N\sigmat_k^2 \sigma^2$ since $\z(m)\sim\CN(0,\sigma^2)$ and $\B$ is fixed. Then, we have
\begin{align}
\PPP \left( | \WW_{l,k}|\geq \sqrt{N} \sigmat_k \sigma \beta \right) \leq 2 e^{-\beta^2}
\label{gaussian_tail}
\end{align}
for any $\beta > 1/\sqrt{2\pi}$~\cite{heckel2018generalized}.
As a consequence, we have
\begin{align*}
\PPP\left( \|\BB^*(\z)\|_{\AA}^*\geq \frac{\lambda}{\eta} \right)
&=\PPP\left\{ \left(\|\BB^*(\z)\|_{\AA}^*\right)^2\geq \frac{\lambda^2}{\eta^2} \right\}\\
&\leq\PPP\left(  \max_{l=0,\ldots,L-1} \sumk \left| \WW_{l,k}  \right|^2  \geq  4 \left(1-\frac{2\pi N}{L} \right) N \sumk \sigmat_k^2 \sigma^2  \log(N)  \right)\\
%&\leq L\PPP\left( \sumk \left| \WW_{l,k}  \right|^2  \geq\sumk \sigmat_k^2  4 \left(1-\frac{2\pi N}{L} \right) N  \sigma^2  \log(N)  \right)\\
&\leq LK\PPP\left( \left| \WW_{l,k}  \right|^2  \geq 4 \sigmat_k^2   \left(1-\frac{2\pi N}{L} \right) N  \sigma^2  \log(N)  \right)\\
&= LK\PPP\left( \left| \WW_{l,k}  \right|  \geq \sqrt{N} \sigmat_k \sigma 2 \sqrt{\left(1-\frac{2\pi N}{L} \right)    \log(N) } \right)
\end{align*}
where the first inequality follows by plugging in~\eqref{Bz_temp1}, $\lambda = 2\eta \sigma \|\B\|_F \sqrt{\log(N)}$, and $\|\B\|_F^2 = N \sumk \sigmat_k^2$. The second  inequality comes from the union bound. By letting $\beta = 2 \sqrt{\left(1-\frac{2\pi N}{L} \right) \log(N)}$ and $L = 8\pi N$, we finally obtain
\begin{align*}
\PPP\left( \|\BB^*(\z)\|_{\AA}^*\geq \frac{\lambda}{\eta} \right)
&\leq 2LKe^{-4\left(1-\frac{2\pi N}{L} \right) \log(N)}\\
&=16\pi KN^{-2} \leq c \frac{1}{N},
\end{align*}
with some numerical constant $c$. Here, the first inequality follows from~\eqref{gaussian_tail}.

\end{proof}

\subsection{Bounding \texorpdfstring{$\frac 1 N \|\xh-\xs\|_2^2$}{$\frac 1 N \|\hat{x}-x^\star\|_2^2$}} \label{bound_error_pr}

Now, we can set the regularization parameter as
$\lambda =2 \eta \sigma \|\B\|_F \sqrt{\log(N)}$ for some constant $\eta\in(1,\infty)$ such that
\begin{align}
\|\BB^*(\z)\|_{\AA}^* \leq \frac{\lambda}{\eta}
\label{Bzle}
\end{align}
holds with probability at least $1-c\frac{1}{N}$, as is shown in Lemma~\ref{Lemma_EPBz}.

With some fundamental computations based on convex analysis, we have the following lemma that  provides optimality conditions for $\Xh$ to be the solution of the atomic norm regularized least-squares problem~\eqref{atomic_denoising}.
\begin{Lemma}
(Optimality Conditions): $\Xh$ is the solution of the atomic norm regularized least-squares problem~\eqref{atomic_denoising} if and only if
\begin{enumerate}
\item $\|\BB^*(\y-\xh)\|_{\AA}^*\leq \lambda$,
\item $\langle  \Xh,\BB^*(\y-\xh) \rangle_{\RRR}=\lambda\|\Xh\|_{\AA}$.
\end{enumerate}
\label{Lemma_optcon}
\end{Lemma}

Define a vector-valued representing measure for the true data matrix $\Xs$ as
\begin{align*}
\bmu(\tau)\triangleq\sumj c_j \h_j\delta (\tau-\tau_j)
\end{align*}
with $\tau\in [0,1),~\|\h_j\|_2=1$, that is, we have
\begin{align*}
\X^{\star}=\sumj c_j \h_j \a(\tau_j)^H=\int_0^1\bmu(\tau) \a(\tau)^Hd\tau.
\end{align*}
%and $\|\Xs\|_{\AA} = \|\bmu\|_{2,\change{TV}}$, where $\|\bmu\|_{2,\change{TV}}$ is an extension of the traditional total variation norm.
Similarly, we can also define a representation measure $\bmuh$ for the recovered data matrix $\Xh$ and represent it as
\begin{align*}
\Xh = \int_0^1\bmuh(\tau) \a(\tau)^Hd\tau.
\end{align*}
Then, a difference measure can be defined as
\begin{align*}
\bnu \triangleq \bmuh-\bmu,
\end{align*}
which implies that we can represent the recovery error as
\begin{align*}
\e \triangleq \xh-\xs = \BB(\Xh-\Xs)=\BB\left(\int_0^1\bnu(\tau) \a(\tau)^Hd\tau\right).
\end{align*}

Define the $j$-th near region corresponding to $\tau_j$ and the far region as \begin{align*}
N_j&\triangleq\{\tau:d(\tau,\tau_j)\leq 0.16/N\},~j = 1, \ldots, J,\\
F&\triangleq[0,1) /\cup_{j=1}^J N_j,  %\label{NFregion},
\end{align*}
where $d(\tau,\tau_j) \triangleq |\tau-\tau_j|$ denotes the wrap-around distance on the unit circle. Define
\begin{align*}
\E \triangleq \Xh-\Xs = \int_0^1\bnu(\tau) \a(\tau)^Hd\tau.
\end{align*}
It follows that $\e = \BB(\E)$ and we can then bound $\|\e\|_2^2$ as
\begin{equation}
\begin{aligned}
\|\e\|_2^2
&=|\langle \BB(\E),\BB(\E) \rangle|\\
%&=|\langle \BB^*\BB(\E),\E \rangle|\\
&=\left|\left\langle \BB^*\BB(\E),\int_0^1\bnu(\tau)\a(\tau)^Hd\tau \right\rangle\right|\\
&=\left| \int_0^1 \left\langle \BB^*\BB(\E),\bnu(\tau)\a(\tau)^H\right\rangle d\tau\right|\\
&= \left| \int_0^1  \bnu(\tau)^H  \BB^*\BB(\E) \a(\tau)  d\tau\right|\\
&=  \left| \int_0^1  \bnu(\tau)^H  \bxi(\tau)  d\tau\right|\\
&\leq  \left| \int_F \bnu(\tau)^H  \bxi(\tau)  d\tau\right|+\sumj  \left| \int_{N_j} \bnu(\tau)^H  \bxi(\tau)  d\tau\right|.
\label{Eb0}
\end{aligned}
\end{equation}
Here, we have defined a vector-valued error function $\bxi(\tau)\triangleq \BB^*\BB(\E) \a(\tau) = \BB^*(\e)\a(\tau)$.

With a little abuse of notation, we define
\begin{align*}
\left\|\bxi(\tau)\right\|_{2,\infty}\triangleq\sup_{\tau\in[0,1)} \left\|\bxi(\tau)\right\|_2. %\label{xi}
\end{align*}
By using the optimality conditions in Lemma~\ref{Lemma_optcon} and the assumption that the bound condition in~\eqref{Bzle} holds, we have
\begin{align*}
\|\bxi(\tau)\|_{2,\infty}&= \sup_{\tau\in[0,1)}  \| \BB^*(\e) \a(\tau)\|_2\\
%&= \sup_{\tau\in[0,1)}  \| \BB^*\BB(\Xh-\Xs) \a(\tau)\|_2\\
%&= \sup_{\tau\in[0,1)}  \| \BB^*(\xh-\xs) \a(\tau)\|_2\\
&= \sup_{\tau\in[0,1)}  \| \BB^*(\xh-\y+\z) \a(\tau)\|_2\\
%&= \sup_{\tau\in[0,1)}  \| \BB^*(\xh-\y ) \a(\tau) +\BB^*(\z) \a(\tau)\|_2\\
&\leq \sup_{\tau\in[0,1)}  \| \BB^*(\xh-\y ) \a(\tau)\|_2 + \sup_{\tau\in[0,1)}  \|\BB^*(\z) \a(\tau)\|_2\\
&=\|\BB^*(\y-\xh )\|_{\AA}^*+\|\BB^*(\z)\|_{\AA}^*\\
& \leq 2\lambda.
%\label{bound_xi}
\end{align*}

To bound the MSE $\frac 1 N \|\e\|_2^2$, we need the following three key lemmas.
\begin{Lemma}
\label{UBE}
Observe that each entry of $\bxi(\tau)= \BB^*\BB(\E) \a(\tau)$ is an order-$N$ trigonometric polynomial. We have
\begin{align*}
\|\e\|_2^2\leq \left\|\bxi(\tau)\right\|_{2,\infty}\left(  \int_F \|\bnu(\tau)\|_2d\tau+I_0+I_1+I_2  \right) %\label{EI}
\end{align*}
with
\begin{align*}
&I_0^j=\left\|  \int_{N_j} \bnu(\tau)d\tau\right\|_2,\\
&I_1^j=N\left\|  \int_{N_j}(\tau-\tau_j) \bnu(\tau)d\tau\right\|_2,\\
&I_2^j=\frac{N^2}{2}\int_{N_j} (\tau-\tau_j)^2\|\bnu(\tau)\|_2d\tau,\\
&I_l=\sumj I_l^j, ~\change{for}~ l=0,1,2.
\end{align*}
\end{Lemma}
The proof of Lemma \ref{UBE} is given in  Appendix \ref{proofUBE}.

\begin{Lemma}
\label{I0I1}
For some numerical constants $C_0$ and $C_1$, we have that
\begin{align}
I_0\leq C_0\left(J \sqrt{\frac{1}{NK} \log\left(\frac{K(J+1)}{\delta}\right)} \|\bxi(\tau)\|_{2,2} +I_2+\int_F \|\bnu(\tau)\|_2 d\tau \right),\label{I02F} \\
I_1\leq C_1\left(J \sqrt{\frac{1}{NK} \log\left(\frac{K(J+1)}{\delta}\right)} \|\bxi(\tau)\|_{2,2} +I_2+\int_F \|\bnu(f)\|_2 df \right), \label{I12F}
\end{align}
hold with probability at least $1-2\delta$ when provided with
%$N \geq CK \log\left(\frac{K(J+1)}{\delta}\right)$
$N\geq C\mu J^2 K \log\left( \frac{NJK}{\delta} \right)$. Here, $\|\cdot\|_{2,2} \triangleq \left( \int_0^1 \|\cdot\|_2^2 d\tau \right)^{\frac 1 2}$ is defined as the $2,2$ norm.
\end{Lemma}
The proof of Lemma \ref{I0I1} is given in Appendix \ref{proofI0I1}.

\begin{Lemma}
\label{FI2}
There exists a numerical constant $C$ such that
\begin{align}
\int_F \|\bnu(\tau)\|_2d\tau +I_2 \leq C J \sqrt{\frac{1}{NK} \log\left(\frac{K(J+1)}{\delta}\right)} \|\bxi(\tau)\|_{2,2}
\label{eqFI2}
\end{align}
holds with probability at least $1-2\delta$ for some sufficiently large $\eta>1$.
%\note{\SL{Note that we don't need $\eta>1$ here now.}}
\end{Lemma}
The proof of Lemma \ref{FI2} is given in Appendix \ref{proofFI2}.

As a consequence of the above three lemmas, we have
\begin{align}
\|\e\|_2^2 \leq & C \lambda J \sqrt{\frac{1}{NK} \log\left(\frac{K(J+1)}{\delta}\right)} \|\bxi(\tau)\|_{2,2} .
\label{exiQ}
\end{align}
Note that
\begin{align*}
 \|\bxi(\tau)\|_{2,2}^2 &= \int_0^1  \|\bxi(\tau)\|_2^2 d\tau
   = \int_0^1  \|\BB^*(\e) \a(\tau)\|_2^2 d\tau\\
& = \langle \BB^*(\e) \a(\tau),\BB^*(\e) \a(\tau)  \rangle
   = \langle \BB^*(\e) ,\BB^*(\e)  \rangle \\
& = \langle \e ,\BB \BB^*(\e)  \rangle = \summ \|\b_m\|_2^2 |\e(m)|^2\\
& \leq \max_{-2M\leq m \leq 2M} \|\b_m\|_2^2 \|\e\|_2^2
\end{align*}
where the fourth equality follows from Parseval's theorem and $\e(m)$ is the $m-$th entry of $\e$. It follows that
\begin{align}
 \|\bxi(\tau)\|_{2,2} \leq \max_{-2M\leq m \leq 2M} \|\b_m\|_2 \|\e\|_2
 \label{xi22}.
\end{align}
Finally, plugging~\eqref{xi22} into~\eqref{exiQ}, we have that
\begin{align*}
\frac 1 N \|\e\|_2^2 & \leq C \lambda^2 \frac 1 N \max_{-2M\leq m \leq 2M} \|\b_m\|_2^2 \frac{J^2}{NK} \log\left(\frac{K(J+1)}{\delta}\right)\\
%& \leq C \eta^2 \sigma^2 \|\B\|_F^2 \frac 1 N \log(N) \max_{-2M\leq m \leq 2M} \|\b_m\|_2^2   \frac {J^2}{NK} \log\left(\frac{K(J+1)}{\delta}\right)\\
& \leq C \eta^2 \sigma^2 \|\B\|_F^2 \max_{-2M\leq m \leq 2M} \|\b_m\|_2^2  \frac{J^2}{N^2K} \log(N) \log\left(\frac{K(J+1)}{\delta}\right)\\
& \leq C \eta^2 \sigma^2 \mu^2 \frac{J^2K}{N} \log(N) \log\left(\frac{JK}{\delta}\right)\\
& \leq C \eta^2 \sigma^2 \mu^2 \frac{J^2K}{N} \log(N) \log\left(JKN\right)
\end{align*}
holds with probability at least $1-cN^{-1}$ when provided with $N\geq C\mu J^2 K \log\left( \frac{NJK}{\delta} \right)$. Here, the last two inequalities follow from the incoherence property~\eqref{incoprop} and by setting $\delta = N^{-1}$.
%and finish the proof of Theorem~\ref{THM_MSE}.

Next, we explain the reason why we use $N\geq C\mu J^2 K \log\left( \frac{NJK}{\delta} \right)$ instead of the lower bound provided in paper~\cite{yang2016super}, which considers the noiseless counterpart of this work.
%$N$ needs to be large enough in order to successfully recover $\Xs$ from $\BB(\Xs)$.
Particularly,~\cite{yang2016super} requires $N$ to satisfy
\begin{align*}
N \geq C\mu JK \log\left( \frac{NJK}{\delta} \right) \log^2\left( \frac{NK}{\delta}  \right)
\end{align*}
if all the $\h_j$ are i.i.d.\ symmetric random samples from the complex unit sphere, namely, $\EEE \h_j \h_j^H = \frac 1 K \I_K$. In order to drop this randomness assumption on $\h_j$ since we never use it in our proof, we make a slight modification of the proof in paper~\cite{yang2016super}. Note that the authors in~\cite{yang2016super} only use the randomness assumption on $\h_j$ in Lemmas~11 and 13. Therefore, we only need to bound $\|\I_1^l(\tau_d)\|_2$ and $\|\I_2^l(\tau_d)\|_2$ in Lemmas~11 and 13 without the randomness assumption on $\h_j$.

In this part, we use the same notation as paper~\cite{yang2016super}. Readers can refer to paper~\cite{yang2016super} for detailed definition of all variables. Inspired by the proof of~\cite[Lemma 5]{chi2016guaranteed}, we have
\begin{align*}
\sup_{\tau_d \in \Omega_{\change{Grid}}}\|\I_1^l(\tau_d)\|_2 &= \sup_{\tau_d \in \Omega_{\change{Grid}}}\|(\V_l(\tau_d) - \EEE \V_l(\tau_d)  )^H \L \h  \|_2\\
&\leq \sup_{\tau_d \in \Omega_{\change{Grid}}} \| \V_l(\tau_d) - \EEE \V_l(\tau_d)  )^H \L \| \|\h\|_2\\
&\leq 4\sqrt{J} \varepsilon_2,
\end{align*}
which is conditioned on $\EE_3 \bigcap \EE_{1,\varepsilon_1}$. Here, $\EE_3$ and $\EE_{1,\varepsilon_1}$ are two events defined in~\cite{yang2016super}. The last inequality follows from $\sup_{\tau_d \in \Omega_{\change{Grid}}} \| \V_l(\tau_d) - \EEE \V_l(\tau_d)  )^H \L \| \leq 4 \varepsilon_2 $ on the event $\EE_3$ and $\|\h\|_2 = \sqrt{J}$ with $\h = [\h_1^H~\cdots~\h_J^H]^H$. Then, we can obtain $\sup_{\tau_d \in \Omega_{\change{Grid}}}\|\I_1^l(\tau_d)\|_2  \leq \varepsilon_4$ by setting $\varepsilon_2 \leq \frac{\varepsilon_4}{4\sqrt{J}}$.

Getting rid of the conditional probability, we have
\begin{align*}
\PPP\left( \sup_{\tau_d \in \Omega_{\change{Grid}}} \| \I_1^l(\tau_d) \|_2 \geq \varepsilon_4,~l=0,1,2,3 \right) \leq 4| \Omega_{\change{Grid}}|\delta_2 + \PPP(\EE^c_{1,\varepsilon_1}).
\end{align*}
It is shown in paper~\cite{yang2016super} that the first term $4| \Omega_{\change{Grid}}|\delta_2 \leq \delta$ and the second term $\PPP(\EE^c_{1,\varepsilon_1}) \leq \delta$ when provided
\begin{align*}
N \geq \frac{640 \cdot 4^{2l} \mu JK}{3\varepsilon_2^2} \log\left(  \frac{4|\Omega_{\change{Grid}}|(2JK+K)}{\delta}  \right)
\end{align*}
and
\begin{align*}
N \geq \frac{80 \mu JK}{\varepsilon_1^2} \log\left( \frac{4JK}{\delta} \right),
\end{align*}
respectively. Thus, for some constant $C$, we have
\begin{align*}
\PPP\left(  \sup_{\tau_d \in \Omega_{\change{Grid}}} \| \I_1^l(\tau_d) \|_2 \geq \varepsilon_4,~l=0,1,2,3   \right) \leq 2\delta
\end{align*}
provided
\begin{align*}
N \geq C\mu JK \max\left\{ \frac{J}{\varepsilon_4^2}\log\left( \frac{|\Omega_{\change{Grid}}|JK}{\delta}  \right),\log\left( \frac{JK}{\delta} \right)   \right\}.
\end{align*}
Note that we set $\varepsilon_1 = \frac{1}{4}$ and absorb all of the constants into one here.

Similarly, conditioned on $\EE_{1,\varepsilon_1}$, we have
\begin{align*}
\sup_{\tau_d \in \Omega_{\change{Grid}}}\|\I_2^l(\tau_d)\|_2 &= \sup_{\tau_d \in \Omega_{\change{Grid}}}\|[\EEE \V_l(\tau_d)]^H (\L - \L' \otimes \I_K) \h  \|_2\\
&\leq \sup_{\tau_d \in \Omega_{\change{Grid}}} \|[\EEE \V_l(\tau_d)]^H (\L - \L' \otimes \I_K)\| \|\h\|_2\\
&\leq C\sqrt{J} \varepsilon_1
\end{align*}
for some numerical constant $C$.
Then, we can obtain $\sup_{\tau_d \in \Omega_{\change{Grid}}}\|\I_2^l(\tau_d)\|_2  \leq \varepsilon_5$ by setting $\varepsilon_1 \leq \frac{\varepsilon_5}{C\sqrt{J}}$.

Getting rid of the conditional probability, we have
\begin{align*}
\PPP\left( \sup_{\tau_d \in \Omega_{\change{Grid}}} \| \I_2^l(\tau_d) \|_2 \geq \varepsilon_5,~l=0,1,2,3 \right) \leq  \PPP(\EE^c_{1,\varepsilon_1}) \leq \delta
\end{align*}
when provided
\begin{align*}
N \geq C\mu J^2K\frac{1}{\varepsilon_5^2} \log\left( \frac{JK}{\delta} \right).
\end{align*}

Now, we have dropped the randomness assumption on $\h_j$ that is used in Lemmas~11 and 13 of paper~\cite{yang2016super}. We can follow the remaining proof of~\cite{yang2016super} and finally get
\begin{align}
N\geq C\mu J^2 K \log\left( \frac{NJK}{\delta} \right).
\label{N_final}
\end{align}
Define $\TT \triangleq \{\tau_1, \tau_2,\cdots,\tau_J \}$ as the true frequency set. The above bound on $N$ can guarantee that the $\ell_2$ norm of the dual polynomial $\QQ(\tau)$ constructed in~\cite{yang2016super} is strictly less than 1 when $\tau \notin \TT$, which is used in the proof of Lemma~\ref{I0I1}.

%\SL{[Mike: Clarify: What conclusion holds when $N$ satisfies this bound? Shuang: Can we talk in person about this when you are back? I feel this part (giving a bound for N) is a little bit abrupt.]}

\section{Conclusion}
\label{conc}

In this work, we recover a signal that consists of a superposition of complex exponentials with unknown waveform modulations from its noisy measurements by solving an atomic norm regularized least-squares problem. We analyze the mean square error (MSE) and provide a theoretical result to bound the MSE in terms of the  noise variance, the total number of uniform samples, the number of true frequencies, and the dimension of the subspace in which the unknown waveform modulations live. Meanwhile, we conduct several numerical experiments to support the theory.
One of the experiments indicates that there is a room to improve the MSE bound and make it scale linearly with the number of true frequencies. We leave this for our future work.

\section*{Acknowledgement}

The authors would like to thank Jonathan Helland at the Colorado School of Mines for some helpful discussions on atomic norm denoising. The authors would also like to thank the anonymous reviewers for their constructive comments and suggestions which greatly improve the quality of this paper.
This work was supported by NSF grant CCF-1409258, NSF grant CCF-1464205, and NSF grant CCF-1704204.

%\SL{[Notes to Mike: Reference [27] is a technique report. I would prefer to keep it here since we did cite this report in our work. I cite this report because it contains the proof details while the published version not. BTW, its published version is also cited as [15]. For reference [20], it was accepted but has not been published yet.]}

\bibliographystyle{ieeetr}
\bibliography{SS}

\newpage
\appendix

%\section{Appendix}
%\label{appe}

\section{Proof of Lemma~\ref{UBE}}
\label{proofUBE}

Let $\u\in\CCC^K$ be any vector with $\|\u\|_2=1$. Define a trigonometric polynomial
\begin{align*}
\gamma (\tau)\triangleq \u^H\bxi(\tau)
\end{align*}
with degree $N$. Then, we have the following two inequalities
\begin{align*}
\sup_{\tau\in[0,1)} |\gamma'(\tau)| &\leq N\sup_{\tau\in[0,1)} |\gamma(\tau)|, \\%\label{Bern1}\\
\sup_{\tau\in[0,1)} |\gamma''(\tau)| &\leq N^2\sup_{\tau\in[0,1)} |\gamma(\tau)|,%\label{Bern2}
\end{align*}
which follow from the Bernstein's inequality for polynomials~\cite{schaeffer1941inequalities}.
As a consequence, we have
\begin{align*}
\sup_{\tau\in[0,1)} \| \bxi'(\tau) \|_2&=\sup_{\tau\in[0,1),\u} |\u^H\bxi'(\tau)|\\
&=\sup_{\tau\in[0,1),\u} |\gamma'(\tau)|\\
&\leq N\sup_{\tau\in[0,1),\u} |\gamma(\tau)|\\
&=N\sup_{\tau\in[0,1),\u} |\u^H\bxi(\tau)|\\
&=N\sup_{\tau\in[0,1)} \| \bxi(\tau) \|_2.
\end{align*}
Therefore, we obtain an upper bound on $\|\bxi'(\tau)\|_{2,\infty}$:
\begin{align*}
\|\bxi'(\tau)\|_{2,\infty} \leq  N\left\|\bxi(\tau)\right\|_{2,\infty}. %\label{Bern11}
\end{align*}
With a similar argument, we also have
\begin{align}
\|\bxi''(\tau)\|_{2,\infty}&\leq N^2\left\|\bxi(\tau)\right\|_{2,\infty}. \label{Bern22}
\end{align}

The Taylor expansion of $\gamma(\tau)$ at $\tau_j$ is
\begin{align*}
\gamma(\tau)=\gamma(\tau_j)+(\tau-\tau_j)\gamma'(\tau_j)+\frac{1}{2}(\tau-\tau_j)^2\gamma''(\widetilde{\tau}_j)
\end{align*}
with some $\widetilde{\tau}_j\in N_j$.
Now, by using the inequality~\eqref{Bern22}, we obtain
\begin{align*}
&\sup_{\u}|\gamma(\tau)-\gamma(\tau_j)-(\tau-\tau_j)\gamma'(\tau_j)|\\
=&\frac{1}{2}(\tau-\tau_j)^2 \sup_{\u}  |\gamma''(\widetilde{\tau}_j)|\\
=&\frac{1}{2}(\tau-\tau_j)^2 \sup_{\u}  |\u^H\bxi''(\widetilde{\tau}_j)|\\
\leq & \frac{1}{2}(\tau-\tau_j)^2 \|\bxi''(\tau)  \|_{2,\infty}\\
\leq & \frac{N^2}{2}(\tau-\tau_j)^2 \|\bxi(\tau)  \|_{2,\infty}.
\end{align*}
Defining a function $\r(\tau)$ as
\begin{align*}
\r(\tau)=\bxi(\tau)-\bxi(\tau_j)-(\tau-\tau_j)\bxi'(\tau_j),
\end{align*}
we note that
\begin{align*}
&\sup_{\u}|\gamma(\tau)-\gamma(\tau_j)-(\tau-\tau_j)\gamma'(\tau_j)|\\
=&\sup_{\u}|\langle\bxi(\tau)-\bxi(\tau_j)-(\tau-\tau_j)\bxi'(\tau_j),\u \rangle|\\
=& \|\bxi(\tau)-\bxi(\tau_j)-(\tau-\tau_j)\bxi'(\tau_j)\|_2\\
=&\|\r(\tau)\|_2.
\end{align*}
Then, we have
\begin{align*}
\|\r(\tau)\|_2\leq \frac{N^2}{2}(\tau-\tau_j)^2 \|\bxi(\tau)  \|_{2,\infty}.
\end{align*}

Now, we can bound the second term in~\eqref{Eb0} as follows
\begin{align*}
&\sumj  \left| \int_{N_j} \bnu(\tau)^H  \bxi(\tau)  d\tau\right|\\
=&\sumj  \left| \int_{N_j} \bnu(\tau)^H  \left[\bxi(\tau_j)+(\tau-\tau_j)\bxi'(\tau_j) +\r(\tau) \right] d\tau\right|\\
\leq & \sumj \left| \int_{N_j} \bnu(\tau)^H \bxi(\tau_j) d\tau\right| + \sumj  \left| \int_{N_j} \bnu(\tau)^H  (\tau-\tau_j)\bxi'(\tau_j) d\tau\right| + \sumj  \left| \int_{N_j} \bnu(\tau)^H \r(\tau) d\tau\right|\\
\leq & \sumj  \|\bxi(\tau_j)\|_2 \left\| \int_{N_j} \bnu(\tau)  d\tau\right\| + \sumj \|\bxi'(\tau_j)\|_2  \left\| \int_{N_j}  (\tau-\tau_j) \bnu(\tau)  d\tau\right\| + \sumj   \int_{N_j} \|\bnu(\tau)\|_2\| \r(\tau)\|_2 d\tau \\
\leq & \left\|\bxi(\tau)\right\|_{2,\infty} \left[ \sumj\left\| \int_{N_j} \bnu(\tau)  d\tau\right\|+   \sumj N \left\| \int_{N_j}  (\tau-\tau_j) \bnu(\tau)  d\tau\right\|  + \sumj \frac{N^2}{2} \int_{N_j} (\tau-\tau_j)^2      \|\bnu(\tau)\|_2  d\tau \right]\\
=&\left\|\bxi(\tau)\right\|_{2,\infty} (I_0 + I_1 + I_2),
\end{align*}
where $I_l, l=0,1,2$ are defined in Lemma~\ref{UBE}. Here, we have plugged in $\bxi(\tau)=\bxi(\tau_j)+(\tau-\tau_j)\bxi'(\tau_j)+\r(\tau)$ to get the first equality. On the other hand, the first term in~\eqref{Eb0} can be bounded as
\begin{align*}
&\left| \int_F \bnu(\tau)^H  \bxi(\tau)  d\tau\right|\\
\leq &  \int_F \left| \bnu(\tau)^H  \bxi(\tau) \right| d\tau\\
\leq & \int_F \|\bnu(\tau)\|_2 \left\|\bxi(\tau)\right\|_2 d\tau\\
\leq &  \left\|\bxi(\tau)\right\|_{2,\infty} \int_F \|\bnu(\tau)\|_2  d\tau
\end{align*}
by using the Cauchy-Schwarz inequality.

Finally, the square error $\|\e\|_2^2$ can be upper bounded as
\begin{align*}
\|\e\|_2^2 \leq \left\|\bxi(\tau)\right\|_{2,\infty} \left( \int_F \|\bnu(\tau)\|_2  d\tau + I_0 + I_1 + I_2   \right)
\end{align*}
and we finish the proof of Lemma~\ref{UBE}.

\section{Proof of Lemma~\ref{I0I1}}
\label{proofI0I1}

To prove Lemma~\ref{I0I1}, we need the following two theorems, which are multiple measurement vector (MMV) random extensions of~\cite[Theorems 4, 5]{tang2015near} and are proved in Appendix~\ref{proof_dual_thm1} and~\ref{proof_dual_thm2}.
\begin{Theorem}
\label{dualstability}
Define a $K$ dimensional unit ball $\HH=\{\h\in \CCC^K: \| \h \|_2=1\}$.
For any $\tau_1, \tau_2, \ldots, \tau_J$ satisfying the minimum separation condition~\eqref{minsepcon}, there exists a dual certificate $\q$ such that the corresponding vector-valued trigonometric polynomial $\QQ(\tau)=\BB^*(\q)\a(\tau)$ satisfies the following properties for some $\q\in \CCC^{N}$ provided that $N\geq C\mu J^2 K \log\left( \frac{NJK}{\delta} \right)$.
\begin{enumerate}
\item For each $j=1,\ldots,J$, $\QQ(\tau_j)=\h_j$ with $\h_j \in \HH$.
\item In each near region $N_j= \{ \tau:d(\tau,\tau_j)<0.16/N  \}$, there exist constants $C_a$ and $C'_a$ such that
\begin{align}
\| \QQ(\tau) \|_2&\leq 1-\frac{C_a}{2}N^2(\tau-\tau_j)^2 \label{SN1}\\
\left\| \h_j-\QQ(\tau) \right\|_2 &\leq \frac{C'_a}{2}N^2(\tau-\tau_j)^2. \label{SN2}
\end{align}
\item In the far region $\tau \in F=[0,1)/\cap_{j=1}^J N_j$, there exists a constant $C_b>0$ such that
\begin{align}
\|\QQ(\tau)\|_2\leq 1-C_b. \label{SF}
\end{align}
\end{enumerate}
\end{Theorem}

\begin{Theorem}
\label{dualstability1}
Define a $K$ dimensional unit ball $\HH=\{\h\in \CCC^K: \| \h \|_2=1\}$.
For any $\tau_1, \tau_2, \ldots, \tau_J$ satisfying the minimum separation condition~\eqref{minsepcon},  there exists a vector-valued trigonometric polynomial $\QQ_1(\tau)=\BB^*(\q_1)\a(\tau)$ that satisfies the following properties for some $\q_1\in \CCC^{N}$ provided that $N\geq C\mu J^2 K \log\left( \frac{NJK}{\delta} \right)$.
\begin{enumerate}
\item In each near region $N_j= \{ \tau:d(\tau,\tau_j)<0.16/N  \}$, there exists a constant $C_a^1$ such that
\begin{align}
\left\| \h_j(\tau-\tau_j)-\QQ_1(
\tau) \right\|_2 &\leq \frac{C_a^1}{2}N(\tau-\tau_j)^2. \label{SN21}
\end{align}
\item In the far region $\tau \in F=[0,1)/\cap_{j=1}^J N_j$, there exists a constant $C_b^1>0$ such that
\begin{align}
\|\QQ_1(\tau)\|_2\leq \frac{C_b^1}{N}. \label{SF1}
\end{align}
\end{enumerate}
\end{Theorem}

Next, we define a dual certificate as follows:
\begin{Definition}\label{dual_cert} (Dual Certificate):
Define a vector $\q\in \CCC^N$ as a dual certificate for $\xs$ if $\q$ makes the corresponding trigonometric polynomial
\begin{align*}
\QQ(\tau) = \BB^*(\q)\a(\tau) = \summ \q(m) \b_m \e_m^H \a(\tau) = \summ \q(m) e^{i 2\pi m \tau} \b_m
\end{align*}
satisfy
\begin{align}
\QQ(\tau_j) &= \h_j, \forall \tau_j \in \TT, \label{dual_cert1} \\
\|\QQ(\tau)\|_2 &< 1, \forall \tau \notin \TT, \label{dual_cert2}
\end{align}
where $\TT \triangleq \{\tau_1, \tau_2,\cdots,\tau_J \}$ is defined as a set containing all the true frequencies.
\end{Definition}

Note that
\begin{equation}
\begin{aligned}
I_0 &= \sumj \left\| \int_{N_j} \v(\tau) d\tau   \right\|_2\\
&=\sumj \int_{N_j} \v(\tau)^Hd\tau \frac{\int_{N_j} \v(\widehat{\tau})d\widehat{\tau}}{\|\int_{N_j} \v(\widehat{\tau})d\widehat{\tau}\|_2}\\
%&=\sumj \int_{N_j} \v(\tau)^H \frac{\int_{N_j} \v(\widehat{\tau})d\widehat{\tau}}{\|\int_{N_j} \v(\widehat{\tau})d\widehat{\tau}\|_2}  d\tau\\
&=\sumj \int_{N_j} \v(\tau)^H \QQ(\tau) d\tau+  \sumj \int_{N_j} \v(\tau)^H \left[     \frac{\int_{N_j} \v(\widehat{\tau})d\widehat{\tau}}{\|\int_{N_j} \v(\widehat{\tau})d\widehat{\tau}\|_2} -\QQ(\tau) \right]  d\tau\\
& \leq \left|\int_0^1 \v(\tau)^H \QQ(\tau) d\tau  \right| + \left|\int_F \v(\tau)^H \QQ(\tau) d\tau  \right|+  \sumj \int_{N_j} \v(\tau)^H \left[     \frac{\int_{N_j} \v(\widehat{\tau})d\widehat{\tau}}{\|\int_{N_j} \v(\widehat{\tau})d\widehat{\tau}\|_2} -\QQ(\tau) \right]  d\tau\\
&\leq \left|\int_0^1 \v(\tau)^H \QQ(\tau) d\tau  \right| + \int_F \|\v(\tau)\|_2 d\tau + C_a'I_2, \label{I02F_tempt1}
\end{aligned}
\end{equation}
where the last inequality follows from $\|\QQ(\tau)\|_2 \leq 1$ and
\begin{align*}
&\sumj \int_{N_j} \v(\tau)^H \left[     \frac{\int_{N_j} \v(\widehat{\tau})d\widehat{\tau}}{\|\int_{N_j} \v(\widehat{\tau})d\widehat{\tau}\|_2} -\QQ(\tau) \right]  d\tau\\
\leq & \sumj \int_{N_j} \|\v(\tau)\|_2 \left\|    \frac{\int_{N_j} \v(\widehat{\tau})d\widehat{\tau}}{\|\int_{N_j} \v(\widehat{\tau})d\widehat{\tau}\|_2} -\QQ(\tau) \right\|_2 d\tau\\
\leq &  \sumj \int_{N_j}  \frac{C'_a}{2}N^2(\tau-\tau_j)^2 \|\v(\tau)\|_2   d\tau\\
= & C'_a I_2
\end{align*}
by using inequality~\eqref{SN2} and the fact that $\frac{\int_{N_j} \v(\widehat{\tau})d\widehat{\tau}}{\|\int_{N_j} \v(\widehat{\tau})d\widehat{\tau}\|_2}$ belongs to $\HH$.

Recall that the linear operator $\BB: \CCC^{K\times N} \rightarrow \CCC^N$ and its adjoint operator $\BB^*: \CCC^N \rightarrow \CCC^{K\times N}$ are defined as in~\eqref{def_B} and~\eqref{def_B*}. Then, we have $\BB \BB^*: \CCC^N \rightarrow \CCC^N$ and $(\BB \BB^*)^{-1}: \CCC^N \rightarrow \CCC^N$ given as
\begin{align*}
\BB\BB^*(\x) &= \change{diag}\left(\left[\|\b_{-2M}\|_2^2,\cdots,\|\b_0\|_2^2,\cdots,\|\b_{2M}\|_2^2 \right]\right)\x,\\
(\BB\BB^*)^{-1}(\x) &= \change{diag}\left(\left[\|\b_{-2M}\|_2^{-2},\cdots,\|\b_0\|_2^{-2},\cdots,\|\b_{2M}\|_2^{-2} \right]\right)\x.
\end{align*}

To get~\eqref{I02F}, we still need to bound the first term in~\eqref{I02F_tempt1}. In particular, we have
\begin{align*}
\left|\int_0^1 \v(\tau)^H \QQ(\tau) d\tau  \right|
=&\left|\int_0^1 \v(\tau)^H \BB^*(\q) \a(\tau) d\tau  \right| \\
%=&\left|\int_0^1 \left\langle   \BB^*(\q) , \v(\tau) \a(\tau)^H \right\rangle d\tau  \right| \\
=&\left| \left\langle   \BB^*(\q) ,  \int_0^1\v(\tau) \a(\tau)^Hd\tau  \right\rangle  \right| \\
=&\left| \left\langle  \q ,  \e  \right\rangle  \right|   =\left| \left\langle   (\BB\BB^*)^{-1}(\q) ,  \BB\BB^*(\e)  \right\rangle  \right|\\
=&\left| \left\langle   \BB^*(\BB\BB^*)^{-1}(\q) ,  \BB^*(\e)  \right\rangle  \right|\\
=&\left| \left\langle   \BB^*(\BB\BB^*)^{-1}(\q) ,  \BB^*\BB(\E)  \right\rangle  \right|.\\
\end{align*}
Define a new polynomial $\QQt(\tau) \triangleq \BB^*(\BB\BB^*)^{-1}(\q)\a(\tau)$ and recall that $\bxi(\tau)= \BB^*\BB(\E) \a(\tau)$. With Parseval's theorem, we obtain
\begin{align*}
\left|\int_0^1 \v(\tau)^H \QQ(\tau) d\tau  \right|
=&\left| \left\langle   \QQt(\tau),  \bxi(\tau)  \right\rangle  \right|  \\
=&\left| \int_0^1     \bxi(\tau)^H \QQt(\tau) d\tau \right|  \\
\leq & \int_0^1 \left\|\bxi(\tau)\right\|_2\| \QQt(\tau)\|_2 d\tau \\
\leq & \left( \int_0^1 \|\bxi(\tau)\|_2^2 d\tau \right)^{\frac 1 2}  \left( \int_0^1 \|\QQt(\tau)\|_2^2 d\tau \right)^{\frac 1 2}\\
\triangleq & \|\bxi(\tau)\|_{2,2} \|\QQt(\tau)\|_{2,2},
\end{align*}
where the last inequality follows from the Cauchy-Schwarz inequality. Here, we define the $2,2$-norm of $\QQt(\tau)$ and $\bxi(\tau)$ as
\begin{align*}
\|\QQt(\tau)\|_{2,2} &\triangleq \left( \int_0^1 \|\QQt(\tau)\|_2^2 d\tau \right)^{\frac 1 2},\\
\|\bxi(\tau)\|_{2,2} &\triangleq \left( \int_0^1 \|\bxi(\tau)\|_2^2 d\tau \right)^{\frac 1 2}.
\end{align*}
It follows that
\begin{align}
\left|\int_0^1 \v(\tau)^H \QQ(\tau) d\tau  \right| \leq  \|\bxi(\tau)\|_{2,2} \|\QQt(\tau)\|_{2,2}.
\label{I02F_first}
\end{align}

The following lemma gives an upper bound for $\|\QQt(\tau)\|_{2,2}$ and is proved in Appendix~\ref{proofQ22}.
\begin{Lemma}
\label{Q22}
Define two events
\begin{align*}
\EE_{\K} &\triangleq  \left\{ \left\| \K-\EEE \K\right\|^2\leq C\frac{J}{NK}\log\left(\frac{K(J+1)}{\delta}\right)\right\},\\
\EE_{\K'} &\triangleq  \left\{ \left\| \K'-\EEE \K'\right\|^2\leq C\frac{JN}{K}\log\left(\frac{K(J+1)}{\delta}\right)\right\}
\end{align*}
with $\PPP(\EE_{\K} )\geq 1-\delta$ and $\PPP(\EE_{\K'} )\geq 1-\delta$.
$\K$ and $\K'$ are two block matrices defined in Appendix~\ref{proofQ22}.
Conditioned on the above two events $\EE_{\K}$ and $\EE_{\K'}$,
the $2,2$ norm of $\QQt(\tau)$ can be bounded as
\begin{align}
\|\QQt(\tau)\|_{2,2} \leq C J \sqrt{\frac{1}{NK} \log\left(\frac{K(J+1)}{\delta}\right)}
\label{eqQ22}
\end{align}
for some numerical constant $C$ provided that $N \geq CK \log\left(\frac{K(J+1)}{\delta}\right)$.
\end{Lemma}

By plugging~\eqref{I02F_first} and~\eqref{eqQ22} into~\eqref{I02F_tempt1}, one can bound $I_0$ as
\begin{align}
I_0 &\leq C_0 \left( J \sqrt{\frac{1}{NK} \log\left(\frac{K(J+1)}{\delta}\right)} \|\bxi(\tau)\|_{2,2}+ \int_F \|\v(\tau)\|_2 d\tau + I_2 \right)
\end{align}
conditioned on the two events $\EE_{\K}$ and $\EE_{\K'}$ and provided that $N \geq CK \log\left(\frac{K(J+1)}{\delta}\right)$.

Similar to~\eqref{I02F_tempt1}, we can divide $I_1$ into the following three parts
\begin{equation}
\begin{aligned}
I_1 &= N\sumj \left\| \int_{N_j}(\tau-\tau_j) \v(\tau) d\tau   \right\|_2\\
&=N\sumj \int_{N_j}(\tau-\tau_j) \v(\tau)^Hd\tau \frac{\int_{N_j} (\widehat{\tau}-\tau_j)  \v(\widehat{\tau})d\widehat{\tau}}{\|\int_{N_j} (\widehat{\tau}-\tau_j) \v(\widehat{\tau})d\widehat{\tau}\|_2}\\
&=N\sumj \int_{N_j} \v(\tau)^H (\tau-\tau_j) \frac{\int_{N_j} (\widehat{\tau}-\tau_j)  \v(\widehat{\tau})d\widehat{\tau}}{\|\int_{N_j} (\widehat{\tau}-\tau_j) \v(\widehat{\tau})d\widehat{\tau}\|_2}d\tau\\
&=N\sumj \int_{N_j} \v(\tau)^H \QQ_1(\tau) d\tau+ N \sumj \int_{N_j} \v(\tau)^H \left[  (\tau-\tau_j) \frac{\int_{N_j} (\widehat{\tau}-\tau_j)  \v(\widehat{\tau})d\widehat{\tau}}{\|\int_{N_j} (\widehat{\tau}-\tau_j) \v(\widehat{\tau})d\widehat{\tau}\|_2} -\QQ_1(\tau) \right]  d\tau\\
& \leq N\!\left|\!\int_0^1\!\!\! \v(\tau)^H \!\QQ_1(\tau) d\tau  \right| \!+\! N\left|\int_F\!\!\! \v(\tau)^H\! \QQ_1(\tau) d\tau  \right| \!+\!  N \!\sumj \!\int_{N_j} \!\!\!\!\v(\tau)^H \!\!\left[ \! (\tau \!-\!\tau_j) \frac{\int_{N_j} (\widehat{\tau}\!-\!\tau_j)  \v(\widehat{\tau})d\widehat{\tau}}{\|\!\int_{N_j} (\widehat{\tau}\!-\!\tau_j) \v(\widehat{\tau})d\widehat{\tau}\|_2} \!-\!\QQ_1(\tau)\! \right] \! d\tau\\
&\leq N\left|\int_0^1 \v(\tau)^H \QQ_1(\tau) d\tau  \right| + C_b^1 \int_F \|\v(\tau)\|_2 d\tau + C_a^1I_2, \label{I12F_tempt1}
\end{aligned}
\end{equation}
where the last inequality follows from~\eqref{SN21} and~\eqref{SF1}. Then, we are left with bounding the first term in~\eqref{I12F_tempt1}.

With a similar trick that used for $I_0$, we have
\begin{align*}
\left|\int_0^1 \v(\tau)^H \QQ_1(\tau) d\tau  \right| \leq \|\bxi(\tau)\|_{2,2} \|\QQt_1(\tau)\|_{2,2}.
\end{align*}
The vector-valued polynomial $\QQt_1(\tau)$ shares the same form as in~\eqref{QKab}, namely,
\begin{align*}
\QQt_1(\tau)=\sumj \Kt_M(\tau-\tau_j)\balpha_j^1 +\sumj \Kt_M'(\tau-\tau_j)\bbeta_j^1,
\end{align*}
with coefficient vectors $\balpha^1 = [{\balpha_1^1}^H~\cdots~{\balpha_J^1}^H]^H$ and $\bbeta^1 = [{\bbeta_1^1}^H~\cdots~{\bbeta_J^1}^H]^H$ satisfying
\begin{align*}
\|\balpha^1\|_2 \leq C_{\alpha^1} \frac{\sqrt{J}}{N},~\|\bbeta^1\|_2 \leq C_{\beta^1} \frac{\sqrt{J}}{N^2},
\end{align*}
which can be verified with a similar trick used in~\eqref{bound_abfinal}.

Similar to Lemma~\ref{Q22}, we can bound $\|\QQt_1(\tau)\|_{2,2}$ as
\begin{align*}
\|\QQt_1(\tau)\|_{2,2} \leq C \frac{J}{N } \sqrt{\frac{1}{NK} \log\left( \frac{K(J+1)}{\delta}  \right)}.
\end{align*}
with probability as least $1-\delta$.
Finally, one can bound $I_1$ as
\begin{align*}
I_1 &\leq C_1 \left( J \sqrt{\frac{1}{NK} \log\left(\frac{K(J+1)}{\delta}\right)} \|\bxi(\tau)\|_{2,2} + \int_F \|\v(\tau)\|_2 d\tau + I_2 \right)
\end{align*}
conditioned on the two events $\EE_{\K}$ and $\EE_{\K'}$ and provided that $N \geq CK \log\left(\frac{K(J+1)}{\delta}\right)$.
Then, we finish the proof of Lemma~\ref{I0I1}.

\section{Proof of Lemma \ref{FI2}}
\label{proofFI2}

Define $\PP_{\TT}(\bnu)$ as the projection of $\bnu(\tau)$ on the true frequency set $\TT \triangleq \{\tau_1, \tau_2,\cdots,\tau_J \}$. Set $\QQ(\tau)$ as the dual polynomial in Thereom~\ref{dualstability}. Denote $\|\cdot\|_{2,\change{TV}}$ as an extension of the traditional TV norm, i.e.,
\begin{equation}
\begin{aligned}
\|\PP_{\TT}(\bnu) \|_{2,\change{TV}}=&\int_0^1 \PP_{\TT}(\bnu^H) \QQ(\tau) d\tau \\
=&\int_{\TT} \bnu(\tau)^H \QQ(\tau) d\tau\\
\leq& \left| \int_0^1 \bnu(\tau)^H \QQ(\tau) d\tau  \right| + \left|\int_{\TT^c} \bnu(\tau)^H \QQ(\tau) d\tau \right|\\
\leq&  C J\! \sqrt{\!\frac{1}{NK} \log \! \left(\!\frac{K(J+1)}{\delta}\!\right)}\|\bxi(\tau)\|_{2,2}\!+\!\!\sum_{\tau_j\in \TT} \left| \int_{N_j/\{\tau_j\}} \!\!\! \bnu(\tau)^H \!\QQ(\tau) d\tau \right| + \left| \int_F\!\! \bnu(\tau)^H\! \QQ(\tau) d\tau \right|,
\label{PTV}
\end{aligned}
\end{equation}
where $\TT^c$ is defined as the complement set of $\TT$ on $[0,1)$. Note that the integration over the far region $F$ can be bounded with
\begin{equation}
\begin{aligned}
\left| \int_F \bnu(\tau)^H \QQ(\tau) d\tau \right|
&\leq \int_F \|\bnu(\tau)\|_2 \|\QQ(\tau)\|_2 d\tau\\
&\leq (1-C_b) \int_F \| \bnu(\tau) \|_2df
\label{PTV_F}
\end{aligned}
\end{equation}
by using~\eqref{SF}.
On the other hand, we can bound the integration over $N_j/\{\tau_j\}$ with
\begin{equation}
\begin{aligned}
\left| \int_{N_j/\{\tau_j\}} \bnu(\tau)^H \QQ(\tau) d\tau \right| &\leq \int_{N_j/\{\tau_j\}} \|\bnu(\tau)\|_2 \|\QQ(\tau)\|_2 d\tau \\
&\leq  \int_{N_j/\{\tau_j\}} \left(1-\frac{1}{2}N^2C_a(\tau-\tau_j)^2  \right) \|\bnu(\tau)\|_2  d\tau \\
&\leq  \int_{N_j/\{\tau_j\}}  \|\bnu(\tau)\|_2  d\tau -C_a I_2^j,
\label{PTV_N}
\end{aligned}
\end{equation}
where the second inequality follows from~\eqref{SN1}. Hence, $\|\PP_{\TT}(\bnu) \|_{2,\change{TV}}$ can be bounded with
\begin{align*}
\|\PP_{\TT}(\bnu) \|_{2,TV}&\leq CJ \sqrt{\frac{1}{NK} \log\left(\frac{K(J+1)}{\delta}\right)} \|\bxi(\tau)\|_{2,2}+\!\!\sum_{\tau_j\in \TT}  \int_{N_j/\{\tau_j\}}\!\!\!\!  \|\bnu(\tau)\|_2  d\tau -C_aI_2+(1\!-\!C_b)\!\! \int_F\!\! \| \bnu(\tau) \|_2d\tau\\
&= C J \sqrt{\frac{1}{NK} \log\left(\frac{K(J+1)}{\delta}\right)}\|\bxi(\tau)\|_{2,2} +\|\PP_{\TT^c}(\bnu) \|_{2,\change{TV}}-C_aI_2-C_b \int_F \| \bnu(\tau) \|_2d\tau
\end{align*}
by plugging~\eqref{PTV_F} and~\eqref{PTV_N} into~\eqref{PTV}. It follows that
\begin{align}
\|\PP_{\TT^c}(\bnu) \|_{2,\change{TV}}-\|\PP_{\TT}(\bnu) \|_{2,\change{TV}}\geq C_aI_2+C_b \int_F \| \bnu(\tau) \|_2d\tau- C J \sqrt{\frac{1}{NK} \log\left(\frac{K(J+1)}{\delta}\right)}\|\bxi(\tau)\|_{2,2}.
\label{TTCL}
\end{align}

As in Lemma~\ref{Lemma_optcon}, denote $\Xh$ as the solution of the atomic norm regularized least-squares problem~\eqref{atomic_denoising}. Then, we have
\begin{align*}
\frac{1}{2}\| \y-\BB(\Xh) \|_2^2+\lambda\| \Xh \|_{\AA}\leq \frac{1}{2}\| \y-\BB(\X^\star) \|_2^2+\lambda\| \X^\star \|_{\AA}.
\end{align*}
By some elementary calculations, we can obtain
\begin{align*}
\lambda\|\Xh\|_{\AA} &\leq \lambda\|\Xs\|_{\AA} + \frac 1 2 \left[\| \y-\BB(\X^\star) \|_2^2- \| \y-\BB(\Xh) \|_2^2\right]\\
&=\lambda\|\Xs\|_{\AA} +  \frac 1 2 \left[ \|\z\|_2^2 - \| \BB(\Xs)+\z-\BB(\Xh) \|_2^2 \right]\\
&= \lambda\|\Xs\|_{\AA} +  \frac 1 2 \left[ 2\left\langle\BB(\Xh-\Xs),\z\right\rangle_{\RRR} - \| \BB(\Xs-\Xh) \|_2^2 \right]\\
&\leq \lambda\|\Xs\|_{\AA} + \left\langle\BB(\Xh-\Xs),\z\right\rangle_{\RRR}\\
&=\lambda\|\Xs\|_{\AA} + \left\langle\e,\z\right\rangle_{\RRR},
\end{align*}
where the first equality follows from $\y = \BB(\Xs)+\z$. Then, we have
\begin{align*}
\|\Xh\|_{\AA} \leq \|\Xs\|_{\AA} + \frac{1}{\lambda} |\langle\e,\z \rangle|,
\end{align*}
which immediately results in
\begin{align}
\|\bmuh\|_{2,\change{TV}} \leq \|\bmu\|_{2,\change{TV}} + \frac{1}{\lambda} |\langle\e,\z \rangle|
\label{mu_2TV}
\end{align}
due to $\|\Xh\|_{\AA} = \|\bmuh\|_{2,\change{TV}}$ and $\|\Xs\|_{\AA} = \|\bmu\|_{2,\change{TV}}$.

Recall that in Lemma~\ref{UBE}, we have shown
\begin{align*}
\|\e\|_2^2 &=\left| \int_0^1  \bnu(\tau)^H \bxi(\tau)  d\tau \right|\\
&\leq  \left\|\bxi(\tau)\right\|_{2,\infty} \left( \int_F \|\bnu(\tau)\|_2  d\tau + I_0 + I_1 + I_2   \right).
\end{align*}
With a similar technique, we can bound the inner product $ |\langle\e,\z \rangle|$ by
\begin{equation}
\begin{aligned}
 |\langle\e,\z \rangle|&= |\langle\BB(\E),\z \rangle| = |\langle \E, \BB^*(\z) \rangle|\\
 &=\left|\left\langle \int_0^1 \bnu(\tau)\a(\tau)^Hd \tau, \BB^*(\z) \right\rangle\right|\\
 &=\left| \int_0^1  \bnu(\tau)^H \BB^*(\z) \a(\tau)  d\tau \right|\\
 &\leq \left\|\BB^*(\z)\a(\tau)\right\|_{2,\infty} \left( \int_F \|\bnu(\tau)\|_2  d\tau + I_0 + I_1 + I_2   \right)\\
 &= \|\BB^*(\z)\|_{\AA}^* \left( \int_F \|\bnu(\tau)\|_2  d\tau + I_0 + I_1 + I_2   \right)\\
 &\leq C\frac{\lambda}{\eta} \left( J \sqrt{\frac{1}{NK} \log\left(\frac{K(J+1)}{\delta}\right)}  \|\bxi(\tau)\|_{2,2} +I_2 +  \int_F \|\bnu(\tau)\|_2  d\tau \right)
 \label{abs_ez}
\end{aligned}
\end{equation}
if $ \|\BB^*(\z)\|_{\AA}^*\leq \frac{\lambda}{\eta}$. Here, we also use~Lemma \ref{I0I1} in the last inequality. Substituting~\eqref{abs_ez} into~\eqref{mu_2TV} leads to
\begin{align*}
& \|\bmu\|_{2,\change{TV}} +  C\frac{1}{\eta} \left(  J \sqrt{\frac{1}{NK} \log\left(\frac{K(J+1)}{\delta}\right)} \|\bxi(\tau)\|_{2,2}  +I_2 +  \int_F \|\bnu(\tau)\|_2  d\tau \right) \\
\geq & \|\bmuh\|_{2,\change{TV}} = \|\bmu+\bnu\|_{2,\change{TV}} \\
%=& \|\bmu +\PP_{\TT}(\bnu) + \PP_{\TT^c}(\bnu) \|_{2,\change{TV}} \\
=&\|\bmu +\PP_{\TT}(\bnu)  \|_{2,\change{TV}}  +\| \PP_{\TT^c}(\bnu) \|_{2,\change{TV}} \\
\geq & \|\bmu\|_{2,\change{TV}} -\|\PP_{\TT}(\bnu)  \|_{2,\change{TV}}  +\| \PP_{\TT^c}(\bnu) \|_{2,\change{TV}},
\end{align*}
which further implies
\begin{align}
\| \PP_{\TT^c}(\bnu) \|_{2,\change{TV}} - \|\PP_{\TT}(\bnu)  \|_{2,\change{TV}} \leq
C \frac{1}{\eta} \left( J \sqrt{\frac{1}{NK} \log\left(\frac{K(J+1)}{\delta}\right)}  \|\bxi(\tau)\|_{2,2} +I_2 +  \int_F \|\bnu(\tau)\|_2  d\tau \right).
\label{TTCU}
\end{align}
Combining~\eqref{TTCL} and~\eqref{TTCU}, we get
\begin{align*}
C(1+\eta^{-1})J \sqrt{\frac{1}{NK} \log\left(\frac{K(J+1)}{\delta}\right)}  \|\bxi(\tau)\|_{2,2} \geq (C_b-C\eta^{-1})\int_F \|\bnu(\tau)\|_2  d\tau+ (C_a-C\eta^{-1})I_2
\end{align*}
Finally, we can obtain~\eqref{eqFI2} with large enough $\eta$ and finish the proof of Lemma~\ref{FI2}.
%\SL{[Mike: Should Theorem 3.1 say ``for sufficiently large $\eta$"? Shuang: Yes, we did mention this in Theorem 3.1.]}

\section{Proof of Theorem \ref{dualstability}}
\label{proof_dual_thm1}

We use the dual polynomial constructed in~\cite{yang2016super}, namely,
\begin{align}
\QQ(\tau)=\sumj \K_M(\tau-\tau_j)\balpha_j +\sumj \K_M'(\tau-\tau_j)\bbeta_j
\label{dual_rand}
\end{align}
with
\begin{align}
\K_M(\tau)\triangleq \frac 1 M \summ  g_M(m) e^{i2\pi \tau m}  \b_m\b_m^H
\label{kernb}
\end{align}
being the random matrix kernel and $\balpha = [\balpha_1^H~\cdots~\balpha_J^H]^H$,  $\bbeta = [\bbeta_1^H~\cdots~\bbeta_J^H]^H$ being the coefficients that are selected such that
\begin{align*}
\QQ(\tau_j) &= \h_j,\\
\QQ'(\tau_j) &= \zero.
\end{align*}
Then, we can ensure that the first and third statements are satisfied due to the construction of this dual polynomial, when $N$ satisfies the lower bound given in~\eqref{N_final}. Then, we are left with proving the second statement.

For all $\tau_j \in \TT$, we have $\|\QQ(\tau_j)\|_2 = 1$ and
\begin{align*}
\frac{d\|\QQ(\tau)\|_2}{d\tau} |_{\tau = \tau_j} = \|\QQ(\tau_j)\|_2^{-1} \langle \QQ'(\tau_j),\QQ(\tau_j) \rangle_{\RRR} = 0
\end{align*}
due to $\QQ'(\tau_j) = \zero$. Furthermore, for all $\tau\in N_j$, we have
\begin{align*}
\frac{d^2\|\QQ(\tau)\|_2}{d\tau^2} &= -\|\QQ(\tau)\|_2^{-3} \langle \QQ'(\tau),\QQ(\tau) \rangle_{\RRR}^2 + \|\QQ(\tau)\|_2^{-1} \left( \|\QQ'(\tau)\|_2^2+\change{Re}\{ \QQ''(\tau)^H \QQ(\tau) \} \right)\\
&\leq -CN^2,
\end{align*}
%\note{\SL{Correct??}}
where the last inequality follows from $ \|\QQ'(\tau)\|_2^2+\change{Re}\{ \QQ''(\tau)^H \QQ(\tau) \} \leq -CN^2$ and $\|\QQ(\tau)\|_2\leq 1$ for some numerical constant $C$~\cite{yang2016super}. The Taylor expansion of $\QQ(\tau)$ at $\tau_j$ gives
\begin{align*}
\|\QQ(\tau)\|_2 &= \|\QQ(\tau_j)\|_2 + (\tau-\tau_j) \frac{d\|\QQ(\tau)\|_2}{d\tau} |_{\tau = \tau_j} + \frac{1}{2} (\tau-\tau_j)^2 \frac{d^2\|\QQ(\tau)\|_2}{d\tau^2} |_{\tau = \taut} \\
&\leq 1-CN^2(\tau-\tau_j)^2
\end{align*}
with some $\taut \in N_j$. Then, setting $C = \frac{1}{2}C_\alpha$, we can obtain~\eqref{SN1}.

Next, we continue to prove~\eqref{SN2}. Recall the dual polynomial constructed in~\cite{yang2014exact}, namely,
\begin{align}
\QQb(\tau) = \sumj \KK_M(\tau-\tau_j) \balphab_j + \sumj \KK_M'(\tau-\tau_j) \bbetab_j,
\label{dual_zai}
\end{align}
where
\begin{align*}
\KK_M(\tau) \triangleq \left[  \frac{\sin(\pi(M+1)\tau)}{(M+1)\sin(\pi\tau)} \right]^4 = \frac{1}{M} \summ  g_M(m) e^{i2\pi \tau m}
\end{align*}
is the squared Fej\'er kernel and $\balphab = [\balphab_1^H~\cdots~\balphab_J^H]^H$,  $\bbetab = [\bbetab_1^H~\cdots~\bbetab_J^H]^H$ are the coefficients that are selected such that
\begin{align*}
\QQb(\tau_j) &= \h_j,\\
\QQb'(\tau_j) &= \zero.
\end{align*}
With the help of the above dual polynomial~\eqref{dual_zai}, we have
\begin{align*}
\| \h_j-\QQ(\tau) \|_2 &= \| (\h_j-\QQb(\tau) ) - (\QQ(\tau) - \QQb(\tau)) \|_2\\
&\leq \| \h_j-\QQb(\tau) \|_2 + \| \QQ(\tau) - \QQb(\tau) \|_2 \\
&\leq C\| \h_j-\QQb(\tau) \|_2
\end{align*}
since $ \| \QQ(\tau) - \QQb(\tau) \|_2$ can be upper bounded with a very small number as is shown in Lemma~15 of paper~\cite{yang2016super}. %\note{\SL{Check!!}}

To obtain~\eqref{SN2}, we next bound $\| \h_j-\QQb(\tau) \|_2$ by following the proof strategies of Lemma~2.5 in~\cite{candes2013super}. Without loss of generality, we consider $\tau_j = 0$ and bound $\| \h_j-\QQb(\tau) \|_2$ in the interval $[0, 0.16/N]$. Define
\begin{align*}
\w(\tau) \triangleq \h_j-\QQb(\tau) = \w_R(\tau) + i \w_I(\tau),
\end{align*}
where $\w_R(\tau)$ and $\w_I(\tau)$ denote the real and imaginary part of $\w(\tau)$, respectively. Then, we have
\begin{align*}
\|\w_R''(\tau)\|_2 &= \left\|  \sumj \KK_M''(\tau-\tau_j) \change{Re}(\balphab_j) + \sumj \KK_M'''(\tau-\tau_j) \change{Re}(\bbetab_j)  \right\|_2\\
&\leq \max_{1\leq j \leq J} \|\balphab_j\|_2 \sum_{\tau_j\in\TT} |\KK''(\tau-\tau_j)| + \max_{1\leq j \leq J} \|\bbetab_j\|_2 \sum_{\tau_j\in\TT} |\KK'''(\tau-\tau_j)| \\
&\leq C \left( |\KK_M''(\tau)|+ \sum_{\tau_j\in \TT\backslash \{0\} } |\KK_M''(\tau-\tau_j)|  \right) + C \frac 1 N \left( |\KK_M'''(\tau)|+ \sum_{\tau_j\in \TT\backslash \{0\} } |\KK_M'''(\tau-\tau_j)|  \right)\\
&\leq CN^2,
\end{align*}
where we have used
\begin{align}
\|\balphab\|_{2,\infty} \triangleq \max_{1\leq j \leq J} \|\balphab_j\|_2 \leq C,~\text{and}~\|\bbetab\|_{2,\infty} \triangleq \max_{1\leq j \leq J} \|\bbetab_j\|_2 \leq C\frac1 N
\label{bound_coef_ab}
\end{align}
for the third line~\cite{yang2014exact, yang2016super}. The last line follows from equation~(2.25) and Lemma~2.7 of paper~\cite{candes2014towards}.

Then, in the interval $[0, 0.16/N]$, due to $\w_R(0) = \w_R'(0) = \w_I(0) = \w_I'(0) = 0$, we have
\begin{align*}
\|\w_R(\tau)\|_2 &= \left\|\w_R(0) + \tau \w_R'(0) + \frac{\tau^2}{2} \w_R''(\taut)\right\|_2 \\
&=\frac{\tau^2}{2} \|\w_R''(\taut)\|_2 \leq CN^2\tau^2
\end{align*}
with some $\taut \in [0, 0.16/N]$. Similarly, we can get
\begin{align*}
\|\w_I(\tau)\|_2 \leq CN^2\tau^2.
\end{align*}
It follows that
\begin{align*}
\|\h_j-\QQb(\tau)\|_2 = \|\w(\tau)\|_2 \leq \|\w_R(\tau)\|_2 + \|\w_I(\tau)\|_2 \leq CN^2\tau^2,
\end{align*}
which implies that
\begin{align*}
\|\h_j-\QQ(\tau)\|_2 \leq CN^2\tau^2
\end{align*}
and we finish the proof.

\section{Proof of Theorem \ref{dualstability1}}
\label{proof_dual_thm2}

In this section, we extend the proof of Lemma 2.7 in~\cite{candes2013super} to prove our Theorem~\ref{dualstability1}.
Define a vector-valued polynomial $\QQ_1(\tau)$ that shares the same form of $\QQ(\tau)$ as in~\eqref{dual_rand}, namely,
\begin{align}
\QQ_1(\tau)=\sumj \K_M(\tau-\tau_j)\balpha_j^1 +\sumj \K_M'(\tau-\tau_j)\bbeta_j^1,
\label{dual_rand1}
\end{align}
where the random matrix kernel $\K_M(\tau)$ is defined in~\eqref{kernb} and the coefficient vectors $\balpha^1 = [{\balpha_1^1}^H~\cdots~{\balpha_J^1}^H]^H$, $\bbeta^1 = [{\bbeta_1^1}^H~\cdots~{\bbeta_J^1}^H]^H$ are selected to satisfy
\begin{align*}
\QQ_1(\tau_j) &= \zero,\\
\QQ_1'(\tau_j) &= \h_j.
\end{align*}

Similar to Appendix~\ref{proof_dual_thm1}, we define another polynomial $\QQb_1(\tau)$ with the squared Fej\'er kernel $\KK_M(\tau)$, namely,
\begin{align}
\QQb_1(\tau) = \sumj \KK_M(\tau-\tau_j) \balphab_j^1 + \sumj \KK_M'(\tau-\tau_j) \bbetab_j^1,
\label{dual_fej1}
\end{align}
where $\balphab^1 = [{\balphab_1^1}^H~\cdots~{\balphab_J^1}^H]^H$,  $\bbetab^1 = [{\bbetab_1^1}^H~\cdots~{\bbetab_J^1}^H]^H$ are the coefficients that are selected such that
\begin{equation}
\begin{aligned}
\QQb_1(\tau_j) &= \zero,\\
\QQb_1'(\tau_j) &= \h_j.
\label{constraint_QQb}
\end{aligned}
\end{equation}
It can be seen that the polynomial $\QQb_1(\tau)$~\eqref{dual_fej1} is to $\QQ_1(\tau)$~\eqref{dual_rand1}  what $\QQb(\tau)$~\eqref{dual_zai} is to $\QQ(\tau)$~\eqref{dual_rand}. Therefore, we can show that  $\| \QQ_1(\tau) - \QQb_1(\tau) \|_2$ is upper bounded with a very small number when provided with~\eqref{N_final} by using a similar strategy to that in paper~\cite{yang2016super}. This further implies that
\begin{align*}
\|\QQ_1(\tau)\|_2 &\leq \|\QQb_1(\tau)\|_2 + \| \QQ_1(\tau) - \QQb_1(\tau) \|_2 \\
&\leq C  \|\QQb_1(\tau)\|_2,\\
\|\h_j(\tau-\tau_j) - \QQ_1(\tau)\|_2 &\leq  \|\h_j(\tau-\tau_j) - \QQb_1(\tau)\|_2 + \| \QQ_1(\tau) - \QQb_1(\tau) \|_2 \\
&\leq C \|\h_j(\tau-\tau_j) - \QQb_1(\tau)\|_2.
\end{align*}
Then, we only need to bound $\|\QQb_1(\tau)\|_2$ and $ \|\h_j(\tau-\tau_j) - \QQb_1(\tau)\|_2$.

Note that the constraints in~\eqref{constraint_QQb} can be expressed in the following matrix form
\begin{align*}
\left[\begin{array}{cc}
\Db_0 \otimes \I_K & \Db_1 \otimes \I_K\\
\Db_1 \otimes \I_K & \Db_2 \otimes \I_K
\end{array}\right]
\left[\begin{array}{c}
\balphab^1\\
\bbetab^1
\end{array}\right] =
\left[\begin{array}{c}
\zero\\
\h
\end{array}\right]
\end{align*}
with $(\Db_l)_{sj} = \KK_M^l(t_s-t_j)$ and $\h = [\h_1^H~\cdots~\h_J^H]^H$. Define
\begin{align*}
\Db = \left[\begin{array}{cc}
\Db_0 & \Db_1 \\
\Db_1 & \Db_2
\end{array}\right],
\end{align*}
It is shown in~\cite{candes2014towards} that $\Db$ is invertible, which implies that $\Db\otimes \I_K$ is also invertible and these coefficient vectors can be expressed as
\begin{align*}
\left[\begin{array}{c}
\balphab^1\\
\bbetab^1
\end{array}\right] &=
\left[\begin{array}{c}
-(\Db_0\otimes \I_K)^{-1}(\Db_1 \otimes \I_K)\\
\I_{JK}
\end{array}\right] \S^{-1} \h \\
&=\left(\left[\begin{array}{c}
-\Db_0^{-1}\Db_1\\
\I_{J}
\end{array}\right]  \otimes \I_K\right) \S^{-1} \h
\end{align*}
with
\begin{align*}
\S &\triangleq \Db_2\otimes \I_K-(\Db_1\otimes \I_K)(\Db_0\otimes \I_K)^{-1}(\Db_1\otimes \I_K)\\
&=(\Db_2-\Db_1\Db_0^{-1}\Db_1)\otimes \I_K \\
&\triangleq \Sb \otimes \I_K.
\end{align*}
Then, the coefficient vectors can be rewritten as
\begin{align*}
\left[\begin{array}{c}
\balphab^1\\
\bbetab^1
\end{array}\right] = \left(\left(\left[\begin{array}{c}
-\Db_0^{-1}\Db_1\\
\I_{J}
\end{array}\right]  \Sb^{-1} \right)\otimes \I_K\right) \h.
\end{align*}

Define $\|\A\|_{\infty} \triangleq \max_{i} \sum_j |a_{ij}|$ as the infinity norm of a matrix $\A$. Using a similar method to that in the proof of Lemma~5.3 in the technical report~\cite{yang2014exact_techrepo}, we can bound the $l_{2,\infty}$ norm of $\balphab^1$ and $\bbetab^1$ as
\begin{equation}
\begin{aligned}
\|\balphab^1\|_{2,\infty} & \triangleq \max_{1\leq j \leq J} \|\balphab_j^1\|_2 \leq \|\Db_0^{-1}\Db_1  \Sb^{-1} \otimes \I_K\|_{\infty} \|\h\|_{2,\infty}\\
&\leq \|\Db_0^{-1}\Db_1  \Sb^{-1} \|_{\infty} \leq C \frac 1N,\\
\|\bbetab^1\|_{2,\infty} & \triangleq \max_{1\leq j \leq J} \|\bbetab_j^1\|_2 \leq \|\Sb^{-1}\|_{\infty} \leq C\frac{1}{N^2},
\label{bound_coef_ab1}
\end{aligned}
\end{equation}
 where we have used the bounds (B.7) and (B.8) in paper~\cite{candes2013super}. It follows that
 \begin{align*}
\|\QQ_1(\tau)\|_2 &\leq C \|\QQb_1(\tau)\|_2\\
&\leq \|\balphab^1\|_{2,\infty} \sum_{\tau_j \in\TT} |\KK_M(\tau-\tau_j)| +  \|\bbetab^1\|_{2,\infty} \sum_{\tau_j \in\TT} |\KK_M'(\tau-\tau_j)| \\
&\leq C \frac{1}{N},
\end{align*}
where the last inequality follows from~\eqref{bound_coef_ab1} and (B.9) in paper~\cite{candes2013super}.

With the same method used to obtain~\eqref{SN2}, we can show that
\begin{align*}
\|\h_j\tau - \QQb_1(\tau)\|_2 \leq CN\tau^2
\end{align*}
when we consider $\tau_j = 0$ without loss of generality. Therefore, we obtain~\eqref{SN21} and finish the proof.

\section{Proof of Lemma \ref{Q22}}
\label{proofQ22}

Recall that the dual certificate $\q\in\CCC^N$ constructed in~\cite{yang2016super} satisfies the two conditions~\eqref{dual_cert1} and~\eqref{dual_cert2} that are required in Definition~\ref{dual_cert} when $N$ satisfies the lower bound given in~\eqref{N_final}. Therefore, we use the optimal $\q\in\CCC^N$ constructed in~\cite{yang2016super}, that is
\begin{align}
\q = \left(\W^H \W\right)^{-1} \left(\sumj \left[ \begin{array}{c}
\a(\tau_j)^H \e_{-2M}\b_{-2M}^H\\
\vdots\\
\a(\tau_j)^H \e_{2M}\b_{2M}^H
\end{array}\right]\balpha_j+\sumj \left[    \begin{array}{c}
i2\pi(-2M)\a(\tau_j)^H \e_{-2M}\b_{-2M}^H\\
\vdots\\
i2\pi(2M)\a(\tau_j)^H \e_{2M}\b_{2M}^H
\end{array}     \right]\bbeta_j   \right).
\label{opt_q}
\end{align}
Here, $\W=\change{diag}([w_{-2M},\cdots,w_0,\cdots,w_{2M}])$ denotes a weighting matrix and
\begin{align}
\left[
\begin{array}{c}
\balpha\\
\bbeta
\end{array}
\right] = \left[ \balpha_1^H~\cdots~\balpha_J^H~\bbeta_1^H~\cdots~\bbeta_J^H \right]^H,~\balpha_j,\bbeta_j\in\CCC^{K}
\label{defab}
\end{align}
are some coefficients that satisfy
\begin{align*}
\A\left(\W^H \W \right)^{-1} \A^H \left[
\begin{array}{c}
\balpha\\
\bbeta
\end{array}
\right]  = \c_h,
\end{align*}
where $\A\in\CCC^{2KJ \times N}$ and $\c_h\in \CCC^{2KJ}$ are given as
\begin{align*}
\A &= \left[\begin{array}{ccc}
\b_{-2M}\e_{-2M}^H \a(\tau_1)& \cdots &\b_{2M}\e_{2M}^H \a(\tau_1) \\
\vdots& \ddots & \vdots\\
\b_{-2M}\e_{-2M}^H \a(\tau_J)& \cdots &\b_{2M}\e_{2M}^H \a(\tau_J)\\
-i2\pi(-2M)\b_{-2M}\e_{-2M}^H \a(\tau_1)& \cdots &-i2\pi(2M)\b_{2M}\e_{2M}^H \a(\tau_1) \\
\vdots& \ddots & \vdots\\
-i2\pi(-2M)\b_{-2M}\e_{-2M}^H \a(\tau_J)& \cdots & -i2\pi(2M)\b_{2M}\e_{2M}^H \a(\tau_J)\\
\end{array}\right],\\
\c_h&=\left[\h_1^H~\cdots~\h_J^H~\zero_{K\times 1}^H~ \cdots~  \zero_{K\times 1}^H  \right]^H.
\end{align*}

Plugging in the optimal dual certificate $\q$ in~\eqref{opt_q}, the polynomial $\QQt(\tau)$ can be represented as
\begin{equation}
\begin{aligned}
\QQt(\tau) &= \BB^*(\BB\BB^*)^{-1}(\q) \a(\tau)\\
%&= \BB^*\left\{(\BB\BB^*)^{-1}\left[ \left(\W^H \W\right)^{-1} \left(\sumj \left[ \begin{array}{c}
%\a(\tau_j)^H \e_{-2M}\b_{-2M}^H\\
%\vdots\\
%\a(\tau_j)^H \e_{2M}\b_{2M}^H
%\end{array}\right]\balpha_j+\sumj \left[    \begin{array}{c}
%i2\pi(-2M)\a(\tau_j)^H \e_{-2M}\b_{-2M}^H\\
%\vdots\\
%i2\pi(2M)\a(\tau_j)^H \e_{2M}\b_{2M}^H
%\end{array}     \right]\bbeta_j   \right)\right]\right\} \a(\tau)\\
&= \BB^*\!\!\left\{\!(\BB\BB^*)^{-1} \!\!\left(\!\sumj \left[\!\! \begin{array}{c}
\frac{1}{w_{-2M}^2}e^{-i2\pi \tau_j(-2M)}\b_{-2M}^H\\
\vdots\\
\frac{1}{w_{2M}^2}e^{-i2\pi \tau_j(2M)}\b_{2M}^H
\end{array}\!\!\right]\!\balpha_j\!+\!\sumj \!\left[ \!\!   \begin{array}{c}
\frac{i2\pi(-2M)}{w_{-2M}^2}e^{-i2\pi \tau_j(-2M)}\b_{-2M}^H\\
\vdots\\
\frac{i2\pi(2M)}{w_{2M}^2}e^{-i2\pi \tau_j(2M)}\b_{2M}^H
\end{array}  \!\!   \right]\!\bbeta_j \!\!  \right)\!\!\right\} \!\a(\tau)\\
&=\BB^*\!\!\left(\!\sumj\!\! \left[\! \begin{array}{c}
\!\!\frac{1}{w_{\!-\!2M}^2 \|\b_{\!-\!2M}\|_2^2} e^{-i2\pi \tau_j(\!-\!2M)} \b_{\!-\!2M}^H \balpha_j\\
\vdots\\
\frac{1}{w_{2M}^2  \|\b_{2M}\|_2^2}e^{-i2\pi \tau_j(2M)} \b_{2M}^H \balpha_j
\end{array} \!\!\! \right]\! \right)\! \a(\tau)  \!+\!  \BB^*\!\!\left(\sumj \!\!\left[ \!\!\!   \begin{array}{c}
\frac{i2\pi(-2M)}{w_{\!-\!2M}^2  \|\b_{\!-\!2M}\|_2^2}e^{-i2\pi \tau_j(\!-\!2M)} \b_{\!-\!2M}^H\bbeta_j \\
\vdots\\
\frac{i2\pi(2M)}{w_{2M}^2  \|\b_{2M}\|_2^2}e^{-i2\pi \tau_j(2M)}\b_{2M}^H \bbeta_j
\end{array}  \!\!  \! \right] \! \right)\! \a(\tau)\\
%&= \summ \sumj \frac{1}{w_{m}^2 \|\b_{m}\|_2^2} e^{-i2\pi \tau_jm}  \b_{m}^H \balpha_j \b_m \e_m^H \a(\tau)+\summ \sumj \frac{i2\pi m}{w_{m}^2 \|\b_{m}\|_2^2} e^{-i2\pi \tau_jm}  \b_{m}^H \balpha_j \b_m \e_m^H \a(\tau)\\
&= \sumj \left(\summ \frac{1}{w_{m}^2 \|\b_{m}\|_2^2} e^{i2\pi (\tau-\tau_j)m}   \b_m \b_{m}^H  \right)\balpha_j + \sumj \left(\summ \frac{i2\pi m}{w_{m}^2 \|\b_{m}\|_2^2} e^{i2\pi (\tau-\tau_j)m}   \b_m \b_{m}^H  \right) \bbeta_j\\
&=\sumj \Kt_M(\tau-\tau_j)\balpha_j +\sumj \Kt_M'(\tau-\tau_j)\bbeta_j
\label{QKab}
\end{aligned}
\end{equation}
with
\begin{align*}
\Kt_M(\tau) &= \summ \frac{1}{w_{m}^2 \|\b_{m}\|_2^2}e^{i2\pi \tau m}\b_m \b_{m}^H,\\
\Kt_M'(\tau) &= \summ \frac{i2\pi m}{w_{m}^2 \|\b_{m}\|_2^2}e^{i2\pi \tau m}\b_m \b_{m}^H.
\end{align*}

As in~\cite{yang2016super}, by setting $w_m = \sqrt{\frac{M}{g_M(m)}}$, we have
\begin{align*}
\Kt_M(\tau) &= \frac 1 M \summ  g_M(m) e^{i2\pi \tau m}  \frac{\b_m}{\|\b_{m}\|_2} \frac{\b_m^H}{\|\b_{m}\|_2},\\
\Kt_M'(\tau) &= \frac 1 M \summ (i2\pi m)  g_M(m) e^{i2\pi \tau m}  \frac{\b_m}{\|\b_{m}\|_2} \frac{\b_m^H}{\|\b_{m}\|_2}.
\end{align*}
%Assume that the entries of $\b_m$ satisfy i.i.d.\ $\CN(0,1)$ and $\frac{\b_m}{\|\b_m\|_2}$ are normalized vectors that satisfy uniform distribution on the unit sphere. Then, we have
Recall that we assume $\frac{\b_m}{\|\b_m\|_2}$ satisfies~the isotropy property~\eqref{isotprop}, namely,
\begin{align*}
\EEE\left( \frac{\b_m}{\|\b_m\|_2} \frac{\b_m^H}{\|\b_m\|_2}\right) = \frac 1 K \I_K,
\end{align*}
which implies that
\begin{align*}
\EEE \Kt_M(\tau) &= \frac{1}{MK} \summ  g_M(m) e^{i2\pi \tau m} \I_K = \frac 1 K \KK_M(\tau) \I_K,\\
\EEE \Kt_M'(\tau) &= \frac{1}{MK} \summ (i2\pi m)  g_M(m) e^{i2\pi \tau m}  \I_K= \frac 1 K \KK_M'(\tau) \I_K,
\end{align*}
where $\KK_M(\tau) \triangleq \frac{1}{M} \summ  g_M(m) e^{i2\pi \tau m} $ is the squared Fej\'er kernel.
It follows that
\begin{align*}
\|\QQt(\tau)\|_{2,2}^2 &= \int_0^1 \|\QQt(\tau)\|_2^2 d\tau\\
&= \int_0^1 \left\|\sumj \Kt_M(\tau-\tau_j)\balpha_j +\sumj \Kt_M'(\tau-\tau_j)\bbeta_j\right\|_2^2 d\tau\\
& = \int_0^1 \| \K \balpha + \K' \bbeta\|_2^2 d\tau,
\end{align*}
where $\K \triangleq [ \Kt_M(\tau-\tau_1) ~\cdots~\Kt_M(\tau-\tau_J) ]$ and $\K' \triangleq [ \Kt_M'(\tau-\tau_1) ~\cdots~\Kt_M'(\tau-\tau_J) ]$ are two block matrices with size $K \times KJ$. $\balpha$ and $\bbeta$ are two coefficient vectors defined in~\eqref{defab}. It follows from~\cite{yang2016super} that
\begin{align}
\left\| \left[ \begin{array}{c}
\balpha\\
\sqrt{|\KK_M''(0)|} \bbeta
\end{array}
\right]\right\|_2 \leq \|\D^{-1}\| \|\c_h\|_2 \leq 2\|(\EEE\D)^{-1}\| \|\c_h\|_2 \leq C\sqrt{J},
\label{bound_ab}
\end{align}
where $\D$ denotes a system matrix and is defined as
\begin{align*}
\D \triangleq \left[\begin{array}{cc}
\D_0 & \frac{1}{\sqrt{|\KK_M''(0)|} } \D_1\\
- \frac{1}{\sqrt{|\KK_M''(0)|} } \D_1 & -\frac{1}{|\KK_M''(0)|} \D_2
\end{array}\right]
\end{align*}
with $[\D_l]_{sj} = \Kt_M^{(l)}(\tau_s-\tau_j)$, $l=0,1,2$.\footnote{Note that we use $\Kt_M^{(l)}(\tau)$ to denote the $l$-th order derivative of $\Kt_M(\tau)$, namely, $\Kt_M^{(0)}(\tau) = \Kt_M(\tau)$, $\Kt_M^{(1)}(\tau) = \Kt_M'(\tau)$, and $\Kt_M^{(2)}(\tau) = \Kt_M''(\tau)$.} To obtain~\eqref{bound_ab}, we have used $\|\D^{-1}\| \leq 2 \|(\EEE\D)^{-1}\|$~\cite[Lemma 7]{yang2016super}, $\|(\EEE\D)^{-1}\| \leq 1.568$~\cite[Lemma 4]{yang2016super} and $\|\c_h\|_2 = \sqrt{J}$.
The inequality in~\eqref{bound_ab} also implies that
\begin{align}
\|\balpha\|_2 \leq C_\alpha \sqrt{J},~\text{and}~\|\bbeta\|_2 \leq C_\beta \frac{\sqrt{J}}{N}
\label{bound_abfinal}
\end{align}
since $|\KK_M''(0)| = \frac{4\pi^2(M^2-1)}{3} \leq CN^2$.
Then, we have
\begin{equation}
\begin{aligned}
\|\QQt(\tau)\|_{2,2}^2 & = \int_0^1 \| \K \balpha + \K' \bbeta\|_2^2 d\tau\\
& \leq \int_0^1 \| \K \balpha \|_2^2 d\tau + \int_0^1 \| \K' \bbeta\|_2^2 d\tau + 2\int_0^1 \| \K \balpha \|_2 \|  \K' \bbeta \|_2 d\tau\\
& \leq \int_0^1 \| \K \|^2 \| \balpha \|_2^2 d\tau + \int_0^1 \| \K'  \|^2 \| \bbeta\|_2^2 d\tau + 2\int_0^1 \| \K \| \| \balpha \|_2 \|  \K'  \| \|\bbeta \|_2 d\tau\\
%& \leq C_\alpha^2 J \int_0^1 \| \K \|^2 d\tau + C_\beta^2 \frac{J}{N^2} \int_0^1 \| \K'  \|^2 d\tau + 2 C_\alpha C_\beta \frac J N \int_0^1 \| \K \|  \|  \K'  \|  d\tau\\
& \leq C_\alpha^2 J \int_0^1 \| \K \|^2 d\tau + C_\beta^2 \frac{J}{N^2} \int_0^1 \| \K'  \|^2 d\tau + 2 C_\alpha C_\beta \frac J N \left(\int_0^1 \| \K \|^2 d\tau\right)^{\frac 1 2} \left( \int_0^1 \|  \K'  \|^2  d\tau \right)^{\frac 1 2},
 \label{normQt21}
\end{aligned}
\end{equation}
where the last inequality follows from~\eqref{bound_abfinal} and the Cauchy-Schwarz inequality.
% while the third inequality is obtained by using the following bounds for coefficient vectors $\balpha$ and $\bbeta$~\cite{yang2016super}
%\begin{align*}
%\max_{1\leq j \leq J} \|\balpha_j\|_2 \leq C_{\alpha},~\max_{1\leq j \leq J} \|\bbeta_j\|_2 \leq \frac{C_{\beta}}{N}.
%\end{align*}
%\note{\SL{What we have is \eqref{bound_coef_ab}, not the above bound.}}

To bound $\|\QQt(\tau)\|_{2,2}^2$, we are left with bounding $\int_0^1  \| \K \|^2 d\tau$ and $\int_0^1  \| \K' \|^2 d\tau$. Note that
\begin{equation}
\begin{aligned}
\int_0^1  \| \K \|^2 d\tau
&=\int_0^1  \left\|(\K - \EEE \K) + \EEE\K\right\|^2 d\tau\\
&\leq \int_0^1  \left\|\K - \EEE \K \right\|^2 d\tau  + \int_0^1  \left\| \EEE\K\right\|^2 d\tau  + 2 \int_0^1  \left\|\K - \EEE \K \right\|   \left\| \EEE\K\right\|  d\tau\\
&\leq \int_0^1  \left\|\K - \EEE \K \right\|^2 d\tau  + \int_0^1  \left\| \EEE\K\right\|^2 d\tau  + 2 \left( \int_0^1  \left\|\K - \EEE \K \right\|^2 d\tau \right)^{\frac 1 2}  \left( \int_0^1  \left\| \EEE\K\right\|^2  d\tau \right)^{\frac{1}{2}},
\label{intKt}
\end{aligned}
\end{equation}
where the last inequality follows from the Cauchy-Schwarz inequality.

Denote $\lambda_{\max}(\X)$ as the maximum eigenvalue of a matrix $\X$. The second term in~\eqref{intKt} can be bounded with
\begin{equation}
\begin{aligned}
\int_0^1  \left\| \EEE\K\right\|^2 d\tau & = \int_0^1  \lambda_{\max}\left( \EEE\K \EEE \K^H \right) d\tau \\
& = \int_0^1  \lambda_{\max}\left( \sumj \EEE \Kt_M(\tau-\tau_j) \EEE \Kt_M(\tau-\tau_j)^H  \right) d\tau \\
& = \int_0^1  \lambda_{\max}\left( \sumj \frac{1}{K^2} |\KK_M(\tau-\tau_j)|^2 \I_K  \right) d\tau \\
& = \frac{1}{K^2} \sumj  \int_0^1 |\KK_M(\tau-\tau_j)|^2  d\tau \\
& = \frac{J}{K^2} \summ \frac{1}{M^2} g_M^2(m)\\
& \leq C \frac{J}{NK^2}
\label{EK2}
\end{aligned}
\end{equation}
for some numerical constant $C$. Here we use Parseval's theorem and $\|g_M\|_{\infty} = \sup_m |g_M(m)| \leq 1$~\cite{tang2013compressed} to get the last equality and inequality, respectively.

Next, we bound $\left\|\K - \EEE \K \right\|$ with the matrix Bernstein inequality~\cite{tropp2015introduction}. Define a set of independent zero mean random matrices $\S_m\in\CCC^{K\times KJ}, m=-2M,\ldots,2M$ with
\begin{align*}
\S_m \triangleq \frac 1 M g_M(m) e^{i 2 \pi \tau m} [e^{-i2\pi\tau_1m}~\cdots~e^{-i2\pi\tau_Jm} ] \otimes \left( \frac{\b_m}{\|\b_{m}\|_2} \frac{\b_m^H}{\|\b_{m}\|_2} - \frac 1 K \I_K \right),
\end{align*}
where ``$\otimes$" denotes the Kronecker product. Then, we have
\begin{align*}
\K-\EEE \K = \summ \S_m.
\end{align*}
To apply the matrix Bernstein inequality, we need to bound the spectral norm  $\|\S_m\| $ and the matrix variance statistic of the sum:
\begin{align*}
 \max\left\{\left\|\summ \EEE\left(\S_m \S_m^H\right) \right\| ,\left\|\summ \EEE\left(\S_m^H \S_m\right) \right\|    \right\},
 \end{align*}
which we tackle separately in the sequel.

Note that we can bound the spectral norm as
\begin{align*}
\|\S_m\| & = \|\S_m \S_m^H\|^{\frac 1 2} = \left\| \sumj \frac{1}{M^2} g_M^2(m) \left( \frac{1}{K^2}\I_K  +  \left( 1-\frac 2 K \right)  \frac{\b_m}{\|\b_{m}\|_2} \frac{\b_m^H}{\|\b_{m}\|_2}  \right) \right\|^{\frac 1 2}\\
& \leq  \frac{\sqrt{J}}{M} \left\| \frac{1}{K^2}\I_K  +  \left( 1-\frac 2 K \right)  \frac{\b_m}{\|\b_{m}\|_2} \frac{\b_m^H}{\|\b_{m}\|_2}  \right\|^{\frac 1 2}\\
& \leq C \frac{\sqrt{J}}{M}
\end{align*}
with some numerical constant $C$. Here, the last inequality follows from
\begin{equation}
\begin{aligned}
 &\left\| \frac{1}{K^2}\I_K  +  \left( 1-\frac 2 K \right)  \frac{\b_m}{\|\b_{m}\|_2} \frac{\b_m^H}{\|\b_{m}\|_2}  \right\|^{\frac 1 2}\\
 =&\begin{cases}
\|1-1\|^{\frac 1 2} = 0,~~&K=1,\\
 \left\| \frac 1 4 \I_K \right\|^{\frac 1 2}, &K = 2,\\
\sqrt{ 1-\frac 2 K} \left\| \frac{\b_m}{\|\b_{m}\|_2} \frac{\b_m^H}{\|\b_{m}\|_2} + \frac{1}{K^2-2K} \I_K \right\|^{\frac 1 2} \leq \sqrt{2}, &K\geq 3.
\label{normcase}
 \end{cases}
\end{aligned}
\end{equation}

On the other hand, we have
\begin{align*}
\left\|\summ \EEE\left(\S_m \S_m^H\right) \right\| & = \left\| \summ \sumj \frac{1}{M^2} g_M^2(m)\left( \frac 1 K-\frac{1}{K^2} \right)\I_K \right\|\\
& = \summ \sumj \frac{1}{M^2} g_M^2(m)\left( \frac 1 K-\frac{1}{K^2} \right) \\
& \leq C \frac{J}{MK}
\end{align*}
and
\begin{align*}
\left\|\summ \EEE\left(\S_m^H \S_m\right) \right\|
& = \left\|\summ \frac{1}{M^2} \left( \frac 1 K-\frac{1}{K^2} \right) g_M^2(m) \E_{m\tau} \otimes \I_K \right\|\\
& \leq C\frac{1}{MK} \|\E_{m\tau}\|  \leq C\frac{1}{MK} \|\E_{m\tau}\|_F\\
& =  C\frac{J}{MK}
\end{align*}
where $\E_{m\tau}$ is a $J\times J$ matrix with the $(k,l)-$th entry being $e^{i2\pi(\tau_k-\tau_l)m}$.
Therefore, the matrix variance statistic of the sum can be bounded with $C\frac{J}{MK}$.
Then, applying the matrix Bernstein inequality~\cite{tropp2015introduction} yields that
\begin{align*}
&\PPP\left\{\left\| \K-\EEE \K \right\|\geq t \right\} \\
\leq & (K+KJ) \change{exp}\left( \frac{-3t^2}{ 8\frac{CJ}{MK} }  \right)\\
\leq & K(J+1) \change{exp}\left( -C\frac{MKt^2}{J}  \right),
\end{align*}
for any $t\in[0,C\frac{\sqrt{J}}{K}]$. Set $t = C\sqrt{\frac{J}{MK}\log\left(\frac{K(J+1)}{\delta}\right) }$, which belongs to the interval $[0,C\frac{\sqrt{J}}{K}]$ if $M\geq CK \log\left(\frac{K(J+1)}{\delta}\right)$. Then, we have
\begin{align*}
\PPP\left\{\left\| \K-\EEE \K \right\|\geq C \sqrt{\frac{J}{MK}\log\left(\frac{K(J+1)}{\delta}\right)} \right\} \leq \delta,
\end{align*}
which immediately suggests that
%\begin{align}
%\PPP\left\{\left\| \K-\EEE \K\right\|^2\leq C\frac{J}{NK}\log\left(\frac{K(J+1)}{\delta}\right) \right\} \geq 1- \delta.
%\label{KEK2}
%\end{align}
the following event
\begin{align*}
\EE_{\K}\triangleq  \left\{ \left\| \K-\EEE \K\right\|^2\leq C\frac{J}{NK}\log\left(\frac{K(J+1)}{\delta}\right)\right\}
\end{align*}
holds with probability at least $1-\delta$ provided that $N\geq CK \log\left(\frac{K(J+1)}{\delta}\right)$.
Conditioned on event $\EE_{\K}$, one can show that
\begin{equation}
\begin{aligned}
\int_0^1  \| \K \|^2 d\tau &\leq \int_0^1  \left\|\K - \EEE \K \right\|^2 d\tau  + \int_0^1  \left\| \EEE\K\right\|^2 d\tau  + 2 \left( \int_0^1  \left\|\K - \EEE \K \right\|^2 d\tau \right)^{\frac 1 2}  \left( \int_0^1  \left\| \EEE\K\right\|^2  d\tau \right)^{\frac{1}{2}}\\
&\leq C\frac{J}{NK}\log\left(\frac{K(J+1)}{\delta}\right)  + C\frac{J}{NK^2} + C \sqrt{ \frac{J}{NK^2} \frac{J}{NK}\log\left(\frac{K(J+1)}{\delta}\right)  }  \\
&\leq C\frac{J}{NK}\log\left(\frac{K(J+1)}{\delta}\right)
\label{intKt_final}
\end{aligned}
\end{equation}
holds by using~\eqref{EK2} provided that $N \geq CK \log\left(\frac{K(J+1)}{\delta}\right)$.

Similar to~\eqref{intKt}, we note that
\begin{equation}
\begin{aligned}
\int_0^1  \| \K' \|^2 d\tau
&=\int_0^1  \left\|(\K' - \EEE \K') + \EEE\K'\right\|^2 d\tau\\
&\leq \int_0^1  \left\|\K' - \EEE \K' \right\|^2 \! d\tau \! +\! \int_0^1 \! \left\| \EEE\K'\right\|^2 \!d\tau \! + \!2 \left( \int_0^1\!  \left\|\K' - \EEE \K' \right\|^2\! d\tau \right)^{\frac 1 2} \!\! \left( \int_0^1 \! \left\| \EEE\K'\right\|^2  d\tau \right)^{\frac{1}{2}},
\label{intKtp}
\end{aligned}
\end{equation}
Using Parseval's identify, we can bound the second term in~\eqref{intKtp} as
\begin{align*}
\int_0^1  \left\| \EEE\K'\right\|^2 d\tau & = \int_0^1  \lambda_{\max}\left( \EEE\K' \EEE \K'^H \right) d\tau \\
& = \int_0^1  \lambda_{\max}\left( \sumj \frac{1}{K^2} |\KK_M'(\tau-\tau_j)|^2 \I_K  \right) d\tau \\
& = \frac{1}{K^2} \sumj  \int_0^1 |\KK_M'(\tau-\tau_j)|^2  d\tau \\
& = \frac{J}{K^2} \summ 4\pi^2 m^2 \frac{1}{M^2} g_M^2(m)\\
& \leq C \frac{JN}{K^2}
%\label{EKp2}
\end{align*}
with some numerical constant $C$.

Now, we bound $\left\|\K' - \EEE \K' \right\|$ with the matrix Bernstein inequality~\cite{tropp2015introduction}. Define a set of independent zero mean random matrices $\S_m'\in\CCC^{K\times KJ}, m=-2M,\ldots,2M$ with
\begin{align*}
\S_m' \triangleq \frac 1 M (i2\pi m) g_M(m) e^{i 2 \pi \tau m} [e^{-i2\pi\tau_1m}~\cdots~e^{-i2\pi\tau_Jm} ] \otimes \left( \frac{\b_m}{\|\b_{m}\|_2} \frac{\b_m^H}{\|\b_{m}\|_2} - \frac 1 K \I_K \right).
\end{align*}
Then, we have
\begin{align*}
\K'-\EEE \K' = \summ \S_m'.
\end{align*}
It can be seen that $\S_m'$ is also the first order derivative of $\S_m$ with respect to $\tau$. We can then bound its spectral norm as
\begin{align*}
\|\S_m'\| & = \|\S_m' \S_m'^H\|^{\frac 1 2} = \left\| \sumj 4\pi^2 m^2 \frac{1}{M^2} g_M^2(m) \left( \frac{1}{K^2}\I_K  +  \left( 1-\frac 2 K \right)  \frac{\b_m}{\|\b_{m}\|_2} \frac{\b_m^H}{\|\b_{m}\|_2}  \right) \right\|^{\frac 1 2}\\
& \leq  C\sqrt{J} \left\| \frac{1}{K^2}\I_K  +  \left( 1-\frac 2 K \right)  \frac{\b_m}{\|\b_{m}\|_2} \frac{\b_m^H}{\|\b_{m}\|_2}  \right\|^{\frac 1 2}\\
& \leq C \sqrt{J}
\end{align*}
with some numerical constant $C$. Here, the last inequality follows from~\eqref{normcase}.

Further, we have
\begin{align*}
\left\|\summ \EEE\left(\S_m' \S_m'^H\right) \right\| & = \left\| \summ \sumj 4\pi^2 m^2 \frac{1}{M^2} g_M^2(m)\left( \frac 1 K-\frac{1}{K^2} \right)\I_K \right\|\\
& = \summ \sumj  4\pi^2 m^2 \frac{1}{M^2} g_M^2(m)\left( \frac 1 K-\frac{1}{K^2} \right) \\
& \leq C \frac{JM}{K}
\end{align*}
and
\begin{align*}
\left\|\summ \EEE\left(\S_m'^H \S_m'\right) \right\|
& = \left\|\summ 4\pi^2 m^2 \frac{1}{M^2} \left( \frac 1 K-\frac{1}{K^2} \right) g_M^2(m) \E_{m\tau} \otimes \I_K \right\|\\
& \leq C\frac{M}{K} \|\E_{m\tau}\|  \leq C\frac{M}{K} \|\E_{m\tau}\|_F\\
& =  C\frac{JM}{K}
\end{align*}
where $\E_{m\tau}$ is a $J\times J$ matrix with the $(k,l)-$th entry being $e^{i2\pi(\tau_k-\tau_l)m}$.
Therefore, the matrix variance statistic of the sum can be bounded with $C\frac{JM}{K}$.
Then, we combine the above bounds and apply the matrix Bernstein inequality to obtain
\begin{align*}
&\PPP\left\{\left\| \K'-\EEE \K' \right\|\geq t \right\} \\
\leq & (K+KJ) \change{exp}\left( \frac{-3t^2}{ 8\frac{CJM}{K} }  \right)\\
\leq & K(J+1) \change{exp}\left( -C\frac{Kt^2}{JM}  \right),
\end{align*}
for any $t\in [0, C\frac{M\sqrt{J}}{K}]$. Set $t = C \sqrt{\frac {JM}{K} \log\left(\frac{K(J+1)}{\delta}\right)}$, which belongs to the interval $[0, C\frac{M\sqrt{J}}{K}]$ if $M\geq C K \log\left(\frac{K(J+1)}{\delta}\right)$. Then, we have
\begin{align*}
\PPP\left\{\left\| \K'-\EEE \K'\right\|\geq C \sqrt{\frac {JM}{ K} \log\left(\frac{K(J+1)}{\delta}\right)} \right\} \leq \delta
\end{align*}
and the following event
\begin{align*}
\EE_{\K'}\triangleq  \left\{ \left\| \K'-\EEE \K'\right\|^2\leq C\frac{JN}{K}\log\left(\frac{K(J+1)}{\delta}\right)\right\}
\end{align*}
holds with probability at least $1-\delta$ provided that $N\geq CK \log\left(\frac{K(J+1)}{\delta}\right)$.
Thus, one can show that
\begin{equation}
\begin{aligned}
\int_0^1  \| \K' \|^2 d\tau &\leq \int_0^1\!  \left\|\K' - \EEE \K' \right\|^2 d\tau \! + \!\int_0^1\!  \left\| \EEE\K'\right\|^2 d\tau \! +\! 2 \left( \int_0^1 \! \left\|\K' - \EEE \K' \right\|^2 d\tau \right)^{\frac 1 2}  \left( \int_0^1 \! \left\| \EEE\K'\right\|^2  d\tau \right)^{\frac{1}{2}}\\
&\leq C \frac{JN}{K}  \log\left(\frac{K(J+1)}{\delta}\right)  + C\frac{JN} {K^2} + 2 \sqrt{\frac{JN} {K^2}  \frac{JN}{K}  \log\left(\frac{K(J+1)}{\delta}\right)  }\\
&\leq C \frac{JN}{K}  \log\left(\frac{K(J+1)}{\delta}\right)
\label{intKtp_final}
\end{aligned}
\end{equation}
holds on the event $\EE_{\K'}$ provided that $N \geq CK \log\left(\frac{K(J+1)}{\delta}\right)$.

Plugging~\eqref{intKt_final} and~\eqref{intKtp_final} into~\eqref{normQt21}, we can bound $\|\QQt(\tau)\|_{2,2}^2 $
with
\begin{equation}
\begin{aligned}
\|\QQt(\tau)\|_{2,2}^2
&\leq C_\alpha^2 J \int_0^1 \| \K \|^2 d\tau + C_\beta^2 \frac{J}{N^2} \int_0^1 \| \K'  \|^2 d\tau + 2 C_\alpha C_\beta \frac J N \left(\int_0^1 \| \K \|^2 d\tau\right)^{\frac 1 2} \left( \int_0^1 \|  \K'  \|^2  d\tau \right)^{\frac 1 2}\\
&\leq C \frac{J^2}{NK}  \log\left(\frac{K(J+1)}{\delta}\right)
\label{normQt21_final}
\end{aligned}
\end{equation}
and finish the proof of Lemma~\ref{Q22} by taking square root on both sides of~\eqref{normQt21_final}.

\end{document}